%% file: arxiv_v2.tex
\documentclass[a4paper]{article}

\usepackage[utf8x]{inputenc}
\usepackage[T1]{fontenc}

\usepackage[a4paper,margin=1in]{geometry}

\usepackage[numbers]{natbib}
\usepackage[usenames,dvipsnames]{xcolor}

\definecolor{DarkRed}{rgb}{0.368,0.097,0.078}

\usepackage[colorlinks = true,
    linkcolor = blue,
    anchorcolor = blue,
    citecolor = blue,
    filecolor = blue,
    urlcolor = DarkRed,
    pagebackref]{hyperref}

\input{header.tex}

\title{Learning Revenue Maximization using Posted Prices for Stochastic Strategic Patient Buyers}

\author{
    Eitan-Hai Mashiah $^\star$ \thanks{School of Computer Science, Tel Aviv University; 
    \texttt{eitanhaimashiah@gmail.com}.}
    \and Idan Attias $^\star$ \thanks{Department of Computer Science, Ben-Gurion University; \texttt{idanatti@post.bgu.ac.il}.} 
    \and Yishay Mansour\thanks{School of Computer Science, Tel Aviv University and Google Research, Tel Aviv; \texttt{mansour.yishay@gmail.com}.}
}

\begin{document}
\maketitle
\def\thefootnote{$\star$}\footnotetext{Equal contribution.}\def\thefootnote{\arabic{footnote}}

\begin{abstract}
We consider a seller faced with buyers which have the ability to delay their decision, which we call patience.
Each buyer's type is composed of value and patience, and it is sampled i.i.d. from a distribution.
The seller, using posted prices, would like to maximize her revenue from selling to the buyer. 
In this paper, we formalize this setting and characterize the resulting Stackelberg equilibrium, where the seller first commits to her strategy, and then the buyers best respond. Following this, we show how to compute both the optimal pure and mixed strategies. 
We then consider a learning setting, where the seller does not have access to the distribution over buyer's types. Our main results are the following.
We derive a sample complexity bound for the learning of an approximate optimal pure strategy, by computing the fat-shattering dimension of this setting.     
Moreover, we provide a general sample complexity bound for the approximate optimal mixed strategy. 
We also consider an online setting and derive a vanishing regret bound with respect to both the optimal pure strategy and the optimal mixed strategy. 
\end{abstract}

\section{Introduction}

Pricing is ubiquitous, and it is the primary means by which sellers and buyers interact. It is no surprise that revenue maximization pricing is the topic of a vast amount of literature in economic theory and algorithmic game theory (see \cite{HartlineBook,NisaRougTardVazi07}). In most of the literature the seller and buyer interact instantaneously, and either a transaction occurs (the buyer purchases an item) or not. We are interested in this work in the case where the buyer can potentially delay the purchase decision, depending on his type. We call such buyers \emph{patient buyers}.

There are many examples of patient buyers in the real world. One example is shipping cost, where there are different costs depending on the duration of the shipping. Normally, same day delivery is more expensive than next day delivery, which is more expensive than two-day delivery, and so on. Another example is an online merchant whose production cost varies with the delivery date. Items that have to be shipped immediately cost more to produce than items that need to be shipped after 10 business days. Another scenario is regarding online merchant who observe that a buyer has a shopping bag that was not purchased. The merchant sometimes offers the buyer a limited time discount on the items in the buyer's shopping bag. The buyer has uncertainty regarding future prices after the discount terminates, the prices might return to the original ones or there may be a new offer with an even larger discount. Our model of patient buyers abstracts this phenomena from the buyer perspective, the ability to prolong the time to receive of the desired item.

Patient buyers were introduced in \citet{feldman2016online} and later studied in \cite{koren2017bandits,KorenLM17}. They presented an adversarial online model where the buyers have a valuation and duration for the purchase, namely the sequence of arrivals is controlled by an adversary. They studied the regret compared to the best fixed price. In this work we consider a stochastic setting, and we study the expected revenue from an optimal sequence of prices (rather than a single fixed price).

Our model of patient buyers can be intuitively described as follows. We have a seller that has an unlimited supply from a single item, and would like to maximize her expected revenue.
Each buyer has a type $(v,w)$ where $v$ is his value for the item and $w\in[\wmax]$ is his patience, where $\wmax$ is the maximum patience.
A buyer of type $(v,w)$ has a value $v$ for the item if it is purchased in the first $w$ time steps from his arrival. The types of the buyers are sampled i.i.d. from a distribution $\D$. The seller proposes a price for the buyer at each time step, and observes whether the buyer bought the item or continued to the next time step.
Initially, the seller commits to her pricing strategy and the buyer best respond to it, i.e., this is a \emph{Stackelberg game} where the seller is the leader and the buyer is the follower.
We consider both the case where the distribution $\D$ is known and the case where $\D$ is unknown.

\subsection{Our contributions}
We initially assume that the buyer's type distribution $\D$ is known and derive the following results.
\begin{itemize}[leftmargin=0.5cm]
\item
We show a separation between the best fixed price, the best pure strategy, which is a fixed sequence of prices,  and the best mixed strategy, which is a distribution over price sequences.
\item
We characterize the optimal pure strategy of the seller and show that the sequence of prices are non-increasing and that the buyers will always buy at the end of them patience, if they decide to buy.
\item
For mixed strategies we characterize the buyer's best response strategy.
\item
We show how to compute efficiently the optimal pure strategy.
For the optimal mixed strategy, we give an algorithm which is exponential in the maximum patience and polynomial in the support of the distribution.
\end{itemize}

We then consider a learning setting, where the seller does not know the distribution $\D$ over buyer's types, but can learn it from samples.
\begin{itemize}[leftmargin=0.5cm]
    \item 
    We derive a sample complexity bound for the learning of an approximate optimal pure strategy, by computing the fat-shattering dimension of the setting and showing that it is linear in the maximum patience of a buyer, i.e., $\wmax$. Using the bound on the fat shattering dimension, we derive an upper bound and a lower bound on the sample complexity of
\[
\cO\paren{\min \left\{
\frac{\wmax}{\eps^2}
\;,\;
\frac{\log(\wmax)}{\eps^3}\right\}
+\frac{1}{\eps^2}\log\frac{1}{\delta}},
\qquad
\Omega \paren{
\frac{\wmax}{\eps}+
\frac{1}{\eps^2}\log\frac{1}{\delta}
}.
\]
    \item 
    We give a general sample complexity bound for the approximate optimal mixed strategy. Our sample bound is 
\[
    \cO\paren{
\frac{{\wmax}^4}{\eps^3}+\frac{\wmax^2}{\eps^2}\log\frac{1}{\delta}}.
\]
    \item
    We consider an online setting with $T$ buyers whose type is drawn i.i.d. from an unknown distribution $\D$. We derive a regret bound with respect to the optimal pure strategy of
    $\widetilde{\cO}\paren{\sqrt{T}\sqrt{\wmax}}$.
    We derive a regret bound with respect to the optimal mixed strategy of
    $
    \widetilde{\cO}\paren{ T^{2/3}\wmax^{4/3}}
    $.
\end{itemize}
We deferred the conclusion, discussion and open problems to Appendix~\ref{app:discussion}.

\subsection{Related work}

\noindent\textbf{The FedEx problem.}
The FedEx problem was presented in \cite{fiat2016fedex} and later studied in \cite{saxena2018menu,devanur2020optimal}.
In the FedEx problem the seller is faced with a buyer which has a varying patience for the duration of the delivery date of a package.
The seller offers the buyer a menu, with a lottery for each possible duration.
The buyer selects one of those lotteries and later pays the realized price of the lottery.
The main issue is that this mechanism maximizes the revenue of the seller, over all incentive-compatible mechanisms.

While the two models are clearly related,
there are a few important differences between the two models.
The main difference is regarding \emph{what} the buyer observes and \emph{when} it observes it.
In the Fed-Ex problem the buyer observes only the menu.
Our setting is more interactive. In each day the buyer first observes the realized price (which is potentially drawn from a distribution) and only then decides if to buy or wait. This implies that the buyer has more information, observing the sequence of prices until the current day, before deciding whether to buy or wait. In contrast, in the Fed-Ex the buyer never observes any realization of prices, except for the lottery it selected.
For example, if the FedEx problem has two lotteries, both uniform $[0,1]$, then the buyer will pay an expected revenue of $1/2$ regardless which lottery he picks. In contrast, in our setting if the seller offers in the first two days a uniform price $[0,1]$, the buyer can decide to buy in the first day if the price is less than $1/2$ and otherwise buy in the second day. This would give an expected revenue of $3/8$.
A minor issue is that we focus on posted prices while the FedEx allow for an arbitrary mechanism.

\noindent{\bf Revenue maximization.}
The seminal work of \citet{myerson1981optimal} derives the optimal mechanism for revenue maximization, and shows that for many distributions it coincides with a sealed bid second price auction with a reserve price. That model allows for single parameter buyers, and does not allow to incorporate the dimension of patience.
The main focus of this paper is on pricing strategies for patient buyers which falls outside that framework.

The work of \citet{kleinberg2003value} derives regret bounds for a seller faced with multiple stochastic buyers. The regret is with respect to the best fixed price. In contrast, we compete with the optimal pure and mixed strategies over sequences of prices, due to our patient buyers. 

\noindent{\bf Repeated interaction between single seller and single buyer.}
The works of \cite{mohri2015revenue,mohri2014revenue,AminRS13,AminRS14,VanuntsD2019}  consider a model of repeated interaction between a single buyer and a single seller. The main issue is that due to the repeated interaction, the buyer has an incentive to lower future prices at the cost of sacrificing current utility.
They define strategic regret and derive near optimal strategic regret bounds for various valuation models, using the fact that the buyer's utility is discounted.
First, the buyer has no patience, at each step he needs to decide if to buy or not.
Second, they consider a single fixed buyer while we consider a distribution over buyer's valuation and patience.
Third, they  use discounting to decay the buyer's utility over time, while in our model the buyer's utility depends only on the paid price.
Lastly, they compare to the best fixed price while we compare to either a pure or mixed  price sequence.

\noindent{\bf Patient buyers.} As mentioned in the introduction, patient buyers were introduced in \citet{feldman2016online} and later studied in \cite{koren2017bandits,KorenLM17}. The focus of those works is on regret minimization with respect to the best fixed price.
They consider an adversarial online model where the buyers have a valuation and duration for the purchase, namely the sequence of arrivals is controlled by an adversary. In this work we consider a stochastic setting, namely, the buyers types are sampled i.i.d. from a distribution. We compare the seller expected revenue to the optimal expected revenue from a pure strategy (a fixed sequence of prices) or mixed strategy (a distribution over price sequences). Clearly, our benchmarks allows for a much higher expected revenue.

\noindent{\bf Learning approximate revenue-maximizing mechanisms}
was initiated by \citet{BalcanBHM08}, using samples to design near optimal revenue-maximizing mechanism. 
\citet{HuangMR18} use i.i.d. samples to derive the optimal sell price.
The works of \cite{morgenstern2015pseudo,MorgensternR16,GonczarowskiN17} study the complexity of learning a near optimal revenue maximizing mechanism.
We differ from all this literature due to the patience of our buyers.

\noindent{\bf Stackelberg games.}
\citet{blum2019computing} derive hardness results for large action Stackelberg games.

\section{Model}

We consider a setting of a single seller and multiple buyers, where the seller has unlimited supply of a single item to sell. The seller observes a sequence of $T$ buyers, and with each buyer she interacts for $\wmax$ steps, in each she offers the buyer a (potentially different) price. Each buyer appears only once, and can purchase the item at most once.

The seller's pricing strategy may be either deterministic, $\p=(p_1,\ldots,p_{\wmax}) \in [0,1]^{\wmax}$, or randomized $ \Pd\in \Delta([0,1]^{\wmax})$. We refer to it as $\textit{pure}$ and $\textit{mixed}$ strategies, respectively. When a pure strategy uses only a single price, we refer to it as a \textit{fixed price}.
We assume no price discrimination, the seller plays the same strategy against each of the buyers. 
We denote a shorthand of a pricing vector by $\p_{1:i}=(p_1,\ldots,p_i)$.

Denote by $\e_1, \ldots, \e_{\wmax}$ the unit vectors of size $\wmax$ and by $\e_0$ the zero vector of size $\wmax$.
Namely, $\e_i = (0,\ldots,0,1,$ $ 0,\ldots,0)$ has $1$ in the $i$-th location, and $\e_0 = (0, \ldots,0)$. Define the buyer's decision whether to purchase the item at step $i$ while observing prices $p_1,\ldots,p_i$ by $\pi^i(\p_{1:i})\mapsto \{\e_i, \cont\}$, for  $i \in [\wmax]$.
The buyer's strategy $\pi_{v,w}= (\pi_{v,w}^1,\ldots,\pi_{v,w}^{\wmax})$ of a buyer with value $v$ and patience $w$, is online and defined as $\pi_{v,w}(\p)=\e_i$ if $i \leq w$ is the first step where $\pi^i_{v,w}(\p_{1:i})=\e_i$ or $\pi_{v,w}(\p)=\e_0$ if no such index $i \leq w$ exists, i.e., the buyer does not purchase the item. The \textit{utility} function of a buyer type $(v,w)$ given pricing $\p$ and decision $\e_i$ is defined as $\utildec_{v,w}(\p,\e_i)=\paren{v - \p\cdot\e_i} \ind{0<i\leq w}$ where $\mathbf{x} \cdot \mathbf{y}$ denotes the scalar product of $\mathbf{x}$ and $\mathbf{y}$. Note that $\utildec_{v,w}(\p,\e_0)=0$. The seller's \textit{revenue} for pricing $\p$ and decision $\e_i$ is defined as $\revdec(\p,\e_i) = \p \cdot \e_i=p_i$.

Define the utility of a buying strategy $\pi_{v,w}$ for a buyer type $(v,w)$, given a selling strategy $\Pd$ by
$$
\utilP(\pi_{v,w})=\Eu{\p \sim \Pd}{\utildec_{v,w}(\p, \pi_{v,w}(\p))}.
$$
The buyer would like to maximize his utility, and select
$\pi^{\star}_{v,w,\Pd} =\argmax_{\pi_{v,w}} \utilP(\pi_{v,w}).$
Define the total revenue of a selling strategy $\Pd$ for distribution $\D$ by $$\rev(\Pd;\D)= \Eun{\p\sim \Pd}{\Eu{(v,w)\sim \D}{\revdec(\p, \pi^{\star}_{v,w,\Pd}(\p))}}.$$
The seller would like to maximize her total revenue, and select $\Pd^{\star} =\argmax_{\Pd} \rev(\Pd;\D).$

\paragraph{Learning.}
In the learning setting the seller does not know the distribution $\D$ over buyer types $\Zc = [0,1] \times [\wmax]$, 
instead she receives an i.i.d. samples from $\D$ in order to learn a selling strategy which maximize her expected revenue.
We define the formal model of revenue learning with patient buyers.

\begin{definition}[Revenue PAC-learning]\label{def:revenue-learning}
For any $(\eps,\delta)\in(0,1)$, the sample complexity of $(\eps,\delta)$-PAC revenue learning with respect to a set of strategies $\Omega$, denoted by $\cM(\eps,\delta,\Omega)$, is defined as the smallest $m\in\N \cup \set{0}$, for which there exists
an algorithm $\cA :\Zc^*\rightarrow \Delta([0,1]^{\wmax})$, such that for any distribution $\D$ over 
$\Zc$, upon receiving a random sample $S\sim\D^m$, with probability $1-\delta$ it holds that,
$$\rev(\cA(S);\D) \geq \max_{\Pd\in \Omega}\rev(\Pd;\D)-\eps.$$
\end{definition}
We consider $\Omega$ to be the set of pure strategies, i.e., $\Omega=[0,1]^{\wmax}$, or 
mixed strategies, i.e., $\Omega=\Delta([0,1]^{\wmax})$.

Our second learning model is in the online setting, where the seller gets to see the sample sequentially instead of receiving the whole sample at once.
The seller is facing a sequence of $T$ buyers of types $z_1,\ldots,z_T$ such that each buyer type $z_t$ is drawn i.i.d. from the unknown distribution $\D$. We assume that the seller interacts with a single buyer at a time, that is, each round of the learning consists of one interaction between the seller and a buyer.

    We denote by $r(\Pd_{t} ; z_{t})$, the revenue of an online learner $\cA: (\Pd_{1:t-1}; z_{1:t-1})\mapsto \Pd_{t}$, at round $t$, given a buyer $z_{t}$.
The regret compared to a set of strategies $\Omega$ of a seller $\cA$ for playing strategies $\Pd_1,\ldots,\Pd_T$, given a sequence of buyer types $z_1,\ldots,z_T$, defined by
$$\regret_T^\Omega\big(\cA;\D\big)=\max_{\Pd^{\star}\in \Omega}\sum_{t=1}^T 
\Eu{z_t\sim \D}{r(\Pd^{\star};z_{t}) - r(\Pd_t;z_{t})}.$$

Similar to the offline setting, we consider both the case that $\Omega$ is the set of pure strategies,  and the case that $\Omega$ is the set of mixed strategies.

\paragraph{Notation.} 
Vectors are denoted by bold lower case letters, e.g., $\p$; we denote historical prices until step $i$ by $\pj{i}=(p_1, \ldots, p_i)$, where $\p_{1:0}$ denotes the null vector; $\D_v$ and $\D_w$ denote the marginal distributions of the buyer's value and patience of distribution $\D$, respectively; $\V$ and $\W$ denote the support of $\D_v$ and $\D_w$, respectively; $S_{\D}\subseteq\V\times\W$ denotes the support of $\D$; $[n]$ denotes the set $\set{1, \ldots, n}$; $\lesssim$ and $\gtrsim$ denote inequalities up to an absolute constant factor.

\section{Optimal pricing: characterization and planning}\label{sec:characterization-and-planning}
In this section we derive basic properties of our model and show how to compute both the optimal pure and mixed strategies given the distribution $\D$.

\noindent\textbf{Product distribution.}
We start by showing that when the distribution $\D$ is a product distribution over values and patience, and the distribution over buyers' values is regular, then the seller cannot outperform the single fixed price, as in \citet{myerson1981optimal} (see \cref{app:charact-warmup}). For this reason we focus on a joint distribution $\D$, and show that there is a separation between the best fixed price and best pure strategy, and also between the best pure strategy and the best mixed strategy (see \cref{app:charact-hierarchy}).

\noindent\textbf{Optimal pure selling strategy.}
In this section, we characterize the optimal pure selling strategies, and use it to compute efficiently an optimal pure strategy.

\begin{restatable}{theorem}{thmoptimalpure}
\label{thm:optimal-pure}
Assume the support of the marginal distribution of the buyer’s value, $\V$, is contained in $[\vmin, \vmax] \subseteq [0,1]$. Then, there exists an optimal non-increasing pure selling strategy using only prices from $[\vmin, \vmax]$. Moreover, if $\V$ is a finite set, there exists an optimal non-increasing pure selling strategy using only prices from $\V$.
\end{restatable}
Intuitively, the existence of an optimal non-increasing pure selling strategy follows since each time the price increases, no buyer would buy at the higher price (since he can buy at the lower price). This implies that we can ``replace'' the higher price by the lower price. Notice that when faced with a sequence of non-increasing prices, the buyer is better off waiting for the last step in his patience window where the price is the lowest.
(See a complete proof in \cref{app:optimal-pure}.) Based on the characterization result, we obtain:

\begin{restatable}{theorem}{thmpureplanning}
\label{thm:pure-planning}
There exists an algorithm 
which produces an optimal pure selling strategy for distributions $\D$ over $\V\times[\wmax]$ where $\V$ is a finite set of values, with running time $\O(|\V|^2\wmax)$.
\end{restatable}

Our algorithm uses a dynamic programming approach, which generates a pricing with non-increasing prices. For non-increasing pricing, the strategic buyer buys at the last step in his patience window, as long as the price is lower than his value. The algorithm takes advantage of this in order to simplify the computation of the seller's revenue (see \cref{app:pure-planning}).
\noindent\textbf{Optimal mixed selling strategy.}
In this section, we characterize the buyer's best-response strategy against a given mixed strategy, and use it to find an optimal mixed selling strategy. We present a simple class of buying strategies, which we call \textit{threshold strategies}. We show that for any mixed selling strategy, there exists a buyer's best response strategy which is a threshold strategy.

\begin{definition}
A buying strategy $\pi$ is a \textit{threshold strategy} if, for any history of prices, there is a threshold
$\theta^i(\p_{1:i})$, such that
the buyer buys at the first step $i \leq w$ in which  the  price is at most the corresponding threshold, i.e., $\theta^i(\p_{1:i})\geq p_i$. 
\end{definition}

\begin{restatable}{theorem}{thmoptimalthresholdbuying}\label{thm:optimal-threshold-buying}
For any mixed selling strategy $\Pd$, there is a threshold buying strategy which is a best response.
\end{restatable}

Intuitively, the threshold at step $i$ is set to the price which makes the buyer ``indifferent'' between buying at step $i$ and continuing to step $i+1$. 
If the offered price is lower, the buyer makes the purchase and if the price is higher the buyer waits (see \cref{app:optimal-threshold-buying}). Using the fact that there always exists a buyer's best response strategy which is a threshold strategy, the following holds: 

\begin{restatable}{theorem}{thmmixedplanning}\label{thm:mixed-planning}
There exists an algorithm which produces an optimal mixed selling strategy for distributions $\D$ over $\V\times[\wmax]$ where $\V$ is a finite set of values, and given a finite set of prices $P$, with running time $\O(\poly(|P|^{\wmax})|V|^{2\wmax}\wmax^2)$. 
\end{restatable}

Similar to the algorithm for finding an optimal pure strategy, we also take a dynamic programming approach here. The main difference is that the set of buyer's types that reach time step $i$ and potentially buy there is not anymore simply the buyers with patience $i$. This creates an intricate dependence between the strategy in steps up to step $i$ and the strategy from step $i$ onward. This is the main reason that the resulting algorithm has a running time exponential in the maximum patience $\wmax$ (see \cref{app:mixed-planning}).

\section{Learning selling strategies}\label{sec:learning}
This section would focus on learning, namely, the seller does not have any a priori information about the distribution $\D$.
First, we consider an offline learning model (see \cref{def:revenue-learning}),
where the seller observes a random sample $S\sim \D^m$ of the buyer's type. 
We would like to understand the sample complexity for the seller to learn an approximate optimal pure strategy and approximate optimal mixed strategy. 
In \cref{subsec:sample-complexity-pure}, we derive upper and lower bounds for the class of pure strategies. 
In order to derive the upper bound 
we compute the fat-shattering of pure strategies, and show that it is linear in the maximum patience of a buyer. 
In \cref{subsec:sample-complexity-mix}, we study the sample complexity of mixed strategies, and derive upper bounds, via learning discrete distributions.
Furthermore, in \cref{subsec:regret}, we consider an online setting with $T$ buyers whose type is drawn i.i.d. from an unknown distribution $\D$. We derive regret bounds with respect to the optimal pure strategy and optimal mixed strategy.

\subsection{Background on learning and notation}\label{subsec:background-leaning}

We use the shorthand of $r(\Pd;z)$ for the revenue of a selling strategy $\Pd$ from a buyer $z$, i.e., $\rev(\Pd; z)= \Eb{\p\sim \Pd}{\revdec(\p, \pi^{\star}_{z,\Pd}(\p))}$, where $\pi^{\star}_{z,\Pd}(\cdot)$ is the buyer's best response when having type $z$.
For a sample of buyers' types $S=\set{z_1,\ldots,z_m}$ define the \textit{empirical revenue} of a selling strategy $\Pd \in \Delta([0,1]^{\wmax})$ with respect to $S$, as
    $$\widehat{\rev}(\Pd;S)=\frac{1}{m}\sum_{z\in S} \rev(\Pd;z).$$

The \textit{empirical revenue maximization} learning algorithm $\ERM$
on sample $S$ with respect to a set of strategies $\Omega$, is defined as
    $$\ERM_{\Omega}(S)\in\argmax_{\Pd\in \Omega} \; \widehat{\rev}(\Pd;S).$$
In this section we consider the following set of strategies. The set of pure strategies, which is denoted by $\pureset=[0,1]^{\wmax}$. 
The set of mixed strategies, denoted by $\mixedset= \Delta([0,1]^{\wmax})$
, and the set of mixed strategies that offers prices from a set $\V$, denoted by $\mixedsetV = \Delta(\V^{\wmax})$.

Recall the definition of the fat-shattering dimension \cite{kearns1994efficient,alon1997scale}.

\begin{definition}
(Fat-shattering dimension)
Let $\F$ be a  class of real-valued functions from input space $\Zc$ and $\gamma > 0$. We say that $S = \{z_1, \ldots, z_m\} \subseteq \Zc$ is $\gamma$-shattered by $\F$ if there exists a witness $\cb = (c_1, \ldots, c_m) \in \R^m$ such that for each $\sigmab = (\sigma_1, \ldots, \sigma_m) \in \{-1, 1\}^m$ there is a function $f_{\sigmab} \in \F$ such that
\[
\forall i \in [m] \; 
\begin{cases}
	f_{\sigmab}(z_i) \geq c_i + \gamma, & \text{if $\sigma_i = 1$}\\
    f_{\sigmab}(z_i) \leq c_i - \gamma, & \text{if $\sigma_i = -1$.}
 \end{cases}
\]
The fat-shattering dimension of $\F$ at scale $\gamma$ is 
the cardinality of the largest set of points in $\Zc$ that can be $\gamma$-shattered by $\F$.
\end{definition}
The Pseudo-dimension of a function class $\F$ \cite{pollard1990empirical,haussler1992decision} can be defined as
$\pdim(\F) = \underset{\gamma \searrow 0}{\lim}\fat_\gamma(\F)$. From the monotonicity of the fat-shattering, it holds that $\pdim(\F)\leq \fat_\gamma(\F)$, for any $\gamma>0$.

The following is a well known uniform convergence theorem for classes with finite fat-shattering for all $\gamma>0$. (The proofs for this section appear in \cref{app:background-learning}).

\begin{restatable}{theorem}{uctheorem}\label{uc-theorem}
Let $\F$ be a function class of real-valued functions mapping from $\Zc$ to $[0,1]$.
For an i.i.d. sample $S = \set{z_1, \ldots, z_m}$ from a distribution $\D$ over $\Zc$, with probability $1- \delta$ it holds that,
\[
\sup_{f \in \F} \abs{\frac{1}{m} \sum_{i=1}^m f(z_i) - \Eu{z \sim \D}{f(z)}} \lesssim \frac{1}{\sqrt{m}} \int_{0}^{\infty} \sqrt{\fatgamma{\F}} \,d\gamma + \sqrt{\frac{\log \frac{1}{\delta}}{m}},
\]
where $\lesssim$ means up to an absolute constant factor.
\end{restatable}

\subsection{Sample complexity of pure selling strategies}\label{subsec:sample-complexity-pure}
Consider the pure selling strategies with non-increasing prices $\p=(p_1,\ldots,p_{\wmax}) \in [0,1]^{\wmax}$, such that $p_1 \geq \cdots \geq p_{\wmax}$, where $\wmax$ is the maximal patience window of all buyers. Recall that by \cref{thm:optimal-pure}, there exist an optimal non-increasing pure strategy, so it suffices to find an approximation of the optimal pure selling strategy in this set of strategies.

Our main result for this section is upper and lower bounds on the sample complexity (\cref{def:revenue-learning}) for learning an approximate optimal pure selling strategy. (The proofs for this section are in \cref{app:sample-complexity-pure}).
\begin{theorem}\label{thm:pure-sample-complexity}
The sample complexity for learning pure selling strategies is
\begin{align*}
&
\cM\paren{\eps,\delta,\Omega^{pure}} = 
\cO\paren{\min \left\{
\frac{\wmax}{\eps^2}
\;,\;
\frac{\log(\wmax)}{\eps^3}\right\}
+\frac{1}{\eps^2}\log\frac{1}{\delta}},
\\
&
\cM\paren{\eps,\delta,\Omega^{pure}} = \Omega \paren{
\frac{\wmax}{\eps}+
\frac{1}{\eps^2}\log\frac{1}{\delta}
}.
\end{align*}
\end{theorem}
\begin{remark}
Note that for a sample of size $m$, when $m\lesssim \wmax^3$ (up to log factors), the error scales roughly as $1/m^{1/3}$ and for $m \gtrsim \wmax^3$ it scales as $\sqrt{\wmax/m}$. 
Moreover, for $\wmax=\O(1)$, our sample complexity bound
is tight.
\end{remark}

The work of \cite{guo2021generalizing} proved an improved sample complexity for product distributions. Our results apply for any distribution.

For obtaining the first upper bound, we compute the fat-shattering dimension for the class of revenues with respect to non-increasing pure strategies, and the claim follows from a uniform convergence argument.
Define the class,
$$\Rpn= \set{z \mapsto r(\p,z):\; \p=(p_1,\ldots,p_{\wmax}) \in [0,1]^{\wmax},p_1 \geq \cdots \geq p_{\wmax}}.$$

\begin{restatable}{lemma}{lemfatgain}\label{lem:fat-gain}
For any $\gamma\in(0,1/4)$, we have $\fatgamma{\Rpn} \in [\wmax, 2\wmax]$. For any $\gamma \in [1/4,1/2]$, we have $\fatgamma{\Rpn} \leq \wmax$. For any $\gamma > 1/2$, we have $\fatgamma{\Rpn} = 0$.
\end{restatable}

We briefly explain how we compute the fat-shattering dimension.
For the upper bound, we first ``project'' the class $\Rpn$ on each patience $w \in [\wmax]$ to obtain $\wmax$ classes with fixed patience.
We show that the fat-shattering of any such projected class is exactly $2$ for $\gamma \in (0,1/4)$, $1$ for $\gamma \in [1/4,1/2]$, and $0$ for $\gamma > 1/2$. We show that this implies 
the appropriate upper bound on the fat-shattering dimension of $\Rpn$.
As for the lower bound, we present a set of $\wmax$ buyer types $(v_i, i)$, with values $v_i$ decreasing with patience $i$. To show that the set is $\gamma$-shattered, we take the witness $(c_1, \ldots, c_{\wmax})$, defined by $c_i = v_i - \gamma$, and prove that for each sequence $(\sigma_1, \ldots, \sigma_{\wmax}) \in \{-1,+1\}^{\wmax}$, there exits a corresponding shattering pricing. The proof is in \cref{app:sample-complexity-pure}, with additional more refined bounds. 

By plugging in the fat-shattering dimension to the uniform convergence bound,
we obtain a bound on the sample complexity.

\begin{restatable}{lemma}{thmfatupper}\label{thm:fat-upper} 
The sample complexity for learning pure selling strategies is
\[ 
\cM(\eps,\delta,\pureset) = \mathcal{O} \left(\frac{\wmax + \log \frac{1}{\delta}}{\varepsilon^2} \right).
\]
\end{restatable}

We proceed to the second upper bound with a better dependence on $\wmax$, albeit a worse dependence on $\eps$. We are doing so by discretizing the set of prices from which the non-increasing strategy is choosing from. The discretization would cause that there are only few changes in prices in the price sequence.
This discretization implicitly implies a shorter horizon $\wmax$, which can be used (implicitly) to derive the improved bound.
The generalization follows from learning a finite class.

\begin{restatable}{lemma}{thmpuresecondupper}\label{thm:pure-second-upper}
The sample complexity for learning pure selling strategies is
$$\cM(\eps,\delta,\pureset) 
=
\cO\paren{
\frac{\log(\wmax)}{\eps^3}+\frac{1}{\eps^2}\log\frac{1}{\delta}}.$$
\end{restatable}

 For the regime where $\wmax \gtrsim 1/\eps$ (up to log factors), the upper bound in \cref{thm:pure-second-upper} is better, and
when $\wmax \lesssim 1/\eps$ the bound in \cref{thm:fat-upper} gives the sample complexity upper bound.

Concerning the lower bounds,
we start with the following simple lower bound,
for a distribution where all buyer types have the same patience. This bound follows from a standard claim on the number of samples needed in order to distinguish between a Bernoulli random variable with parameter $1/2$ and a Bernoulli random variable with parameter $1/2+\eps$ (e.g., \citet[Lemma 2.7]{slivkins2019introduction}).

\begin{restatable}{lemma}{thmlowerboundone}\label{thm:lower-bound-one}
The sample complexity for learning pure selling strategies is
$$\cM\paren{\eps,\delta,\pureset} = \Omega \paren{\frac{1}{\eps^2}\log\frac{1}{\delta}}.$$
\end{restatable}

We prove a second lower bound for a set of natural distributions, where buyer types with larger patience window have strictly lower values, i.e., for any $(v,w)$ and $(v',w')$ in the support of $\D$, if $w'>w$ then $v'<v$.
In order to prove this lower bound, we first claim that for such distributions, the Bayes optimal is a pure selling strategy which offers prices at each step $i$, only from values of buyer types with patience $i$. Note that for such distributions, it suffices to choose the optimal price for each step, independently from the other steps.
We then define a family of distributions, such that in order to find the optimal price at step $i$, the learner should see at least $\approx \frac{1}{\eps}$ samples from patience $i$. This eventually leads to the following bound.

\begin{restatable}{lemma}{thmlowerboundtwo}\label{thm:lower-bound-two}
Let $\wmax\geq 2$. The sample complexity for learning pure selling strategies is
$$\cM\paren{\eps,\delta,\pureset} = \Omega \paren{\frac{\wmax}{\eps}}.$$
\end{restatable}

\cref{thm:pure-sample-complexity} follows immediately by combining \cref{thm:fat-upper,thm:pure-second-upper} for the upper bound, and \cref{thm:lower-bound-one,thm:lower-bound-two} for the lower bound.

\subsection{Sample complexity of mixed selling strategies}\label{subsec:sample-complexity-mix}
In this section, we address the challenging case of learning an approximate optimal mixed strategy from samples.
Initially we will assume that the support of $\D$ is finite and use the empirical distribution to approximate it. Later we generalize the result to arbitrary support, using a discretization of the support. (The proofs for this Section are in \cref{app:sample-complexity-mix}).

Let $\D$ be a distribution over buyer types $[0,1]\times[\wmax]$, where the size of the set of distinct values $\V$ is at most $k$, hence, the support of $\D$ is at most $k\wmax$.
The sample complexity of learning discrete distributions over a known domain of size $k\wmax$, with respect to the total variation distance, is $\Theta\left(\frac{k\wmax+\log\frac{1}{\delta}}{\eps^2} \right)$.\footnote{This is also known as Bretagnolle Huber-Carol inequality, for a proof, see  \citet[Theorem 1]{canonne2020short}.}

Using the sample, we learn an approximation $\hat{\D}$ to $\D$ with a total variation distance of at most $\eps$.
We then use the distribution $\hat{\D}$ to derive an approximate optimal mixed strategy.
We conclude an upper bound on the sample complexity of the set $\mixedsetV$. 

\begin{restatable}{theorem}{thmsamplecomplexitymixone}\label{thm:sample-complexity-mix-one}
The sample complexity for learning mixed selling strategies $\mixedsetV$ is
$$
\cM\paren{\eps,\delta,\mixedsetV}
=
\cO\paren{
\frac{k\wmax}{\eps^2}+\frac{1}{\eps^2}\log\frac{1}{\delta}}.$$
\end{restatable}

When we have large set $\V$ we can discretize it using the parameter $\eps$.
Namely, given $\eps>0$, we discretize the set of values to multiples of $\eps$, resulting in $k=1/\eps$ distinct values.
We have shown that the error incurred in the discretization process is at most $\eps\wmax$ (\cref{lem:eps-w-optimal-mixed}).
Now, using \cref{thm:sample-complexity-mix-one}, with an accuracy parameter $\eps/\wmax$, we derive the following theorem.

\begin{restatable}{theorem}{thmsamplecomplexitymixtwo}\label{thm:sample-complexity-mix-two}
The sample complexity for learning mixed selling strategies $\Omega^{\text{mixed}}$ is
$$
\cM\paren{\eps,\delta,\Omega^{\text{mixed}}}
=
\cO\paren{
\frac{{\wmax}^4}{\eps^3}+\frac{\wmax^2}{\eps^2}\log\frac{1}{\delta}}.$$
\end{restatable}
\subsection{Regret minimization}\label{subsec:regret}
In this section, we address the online setting, where the buyers arrive in an online fashion, and the seller needs to adjust her strategy. This is a stochastic online setting, where the buyers' types are sampled from an unknown distribution $\D$. The goal of the seller is to minimize the regret w.r.t. a given set of strategies. We naturally consider both the case of pure strategies and the case of mixed strategies.

\paragraph{Online model.}
In the online setting, the seller is facing a sequence of $T$ buyers $z_1,\ldots,z_T$ such that each buyer type $z_t$ is drawn i.i.d. from an unknown distribution $\D$.
The seller interacts with only a single buyer at a time, that is, each round of the learning consists of one interaction between the seller and a buyer. At the end of the interaction the seller observes the buyer type (regardless of the outcome).
We prove regret minimization results with respect to the selling strategies: (1) pure selling strategies, i.e., $\pureset$, (2) mixed selling strategies with prices in $\V$, i.e., $\mixedsetV$, and (3) general mixed strategies, i.e., $\Omega^{\text{mixed}}$. (The proofs for this section are in \cref{app:regret}).

We start with the case of comparing to pure strategies. The main idea is to keep the observed buyers' types. When buyer $i$ arrives, there are already $i-1$ observed buyer types. The seller would use the historical observed buyers' types to implicitly learn the distribution $\D$ to a certain accuracy. More explicitly, the seller would invoke an ERM oracle that would select the best empirical strategy on the historical observations. In order to minimize the number of calls to the ERM oracle, we invoke the ERM oracle only for buyers $t$ which are a power of two, i.e., $t=2^i$ for some integer $i$. The difference between the three setting comes from the different convergence rates that we derived for each setting in previous sections.

For the pure strategy setting we use \cref{thm:pure-sample-complexity} to derive the following regret bound.
\begin{restatable}{theorem}{thmregretpure}\label{thm:regret-pure}
There exists an algorithm $\cA$, such that for any distribution $D$, with probability $1-\delta$, 
$$
\regret_T^{\pureset} \big(\cA;\D\big)  = 
\begin{cases}
\cO\paren{\sqrt{T}\sqrt{\wmax}+\sqrt{T}\sqrt{\log\frac{\log T}{\delta}}}, &  T>\frac{{\wmax}^3}{\log^2\wmax} \\
\cO\paren{ T^{2/3}\log^{1/3}(\wmax) +\sqrt{T}\sqrt{\log\frac{\log T}{\delta}}}, & T \leq \frac{{\wmax}^3}{\log^2\wmax},
\end{cases}
$$
where $\cA$ makes $\log T$ calls to an $\ERM$ oracle.
\end{restatable}
The bound follows from \cref{thm:pure-sample-complexity} by setting the accuracy $\eps_m$ as a function of the sample size $m$. For small sample size $m$ we have that $\eps_m\approx \O((\log\wmax/m)^{1/3})$, and for large sample size we have that $\eps_m\approx \O(\sqrt{\wmax/m})$. The regret is now $\sum_{m=1}^T \eps_m$.

For the mixed strategy setting with a limited price set, we use \cref{thm:sample-complexity-mix-one} to derive the following regret bound.
\begin{restatable}{theorem}{thmregretmixone}\label{thm:regret-mix-one}
There exists an algorithm $\cA$, such that for any distribution $D$, with probability $1-\delta$,\begin{align*}
\regret_T^{\mixedsetV}\big(\cA;\D\big)
=
\cO\paren{ \sqrt{T}\sqrt{|\V|\cdot\wmax} +\sqrt{T}\sqrt{\log\frac{\log T}{\delta}}},
\end{align*}
where $\cA$ makes $\log T$ calls to an $\ERM$ oracle.
\end{restatable}
The bound follows from \cref{thm:sample-complexity-mix-one} by setting the accuracy $\eps_m$ as a function of the sample size $m$. Here we have that $\eps_m\approx \O(\sqrt{|\V|\wmax/m})$. Again, the regret is now $\sum_{m=1}^T \eps_m\approx \O(\sqrt{\wmax |\V| T})$.

For the setting with general mixed strategies we use \cref{thm:sample-complexity-mix-two} to derive the following regret bound.
\begin{restatable}{theorem}{thmregretmixtwo}\label{thm:regret-mix-two}
There exists an algorithm $\cA$, such that for any distribution $D$, with probability $1-\delta$,
\begin{align*}
\regret_T^{\Omega^{\text{mixed}}}\big(\cA;\D\big)
=
\cO\paren{ T^{2/3}\wmax^{4/3} +\sqrt{T}\wmax\sqrt{\log\frac{\log T}{\delta}}},
\end{align*}
where $\cA$ makes $\log T$ calls to an $\ERM$ oracle.
\end{restatable}
We now have the freedom to select the set $\V$ as to minimize the regret.
We use a distritization of roughly $T^{1/3}$, which results in the regret bound of order $T^{2/3}$. (The discretization is implicit in the proof of \cref{thm:sample-complexity-mix-two}.) 
The bound follows from \cref{thm:sample-complexity-mix-two} by setting the accuracy $\eps_m$ as a function of the sample size $m$. Here we have that $\eps_m\approx \O(\wmax^{4/3}/m^{1/3})$. Again, the regret is now $\sum_{m=1}^T \eps_m\approx \O(T^{2/3}\wmax^{4/3})$.

\section{Conclusion, discussion and open problems}
\label{app:discussion}

The main focus of this work is on patient buyers, which can delay their purchasing decision.
We presented a new stochastic model where each buyer's type is composed from a value and a patience, and a seller 
who posts prices and would like to maximize her revenue. 
We formalize this setting as a Stackelberg game between a leader (the seller) and a follower (the buyer).

Unlike much of the previous works, our focus is on a sequence of prices rather than a single fixed price. For this end, we show a separation between the best fixed price, the best pure strategy, which is a fixed sequence of prices,  and the best mixed strategy, which is a distribution over price sequences.

We characterize the optimal pure strategy of the seller and show that the sequence of prices are non-increasing and that the buyers will always buy at the end of them patience, if they decide to buy. We also give an efficient algorithm to compute the optimal pure selling strategy. 
We derive a sample complexity bound for the learning of an approximate optimal pure strategy which is polynomial in $\wmax,1/\eps$ and $\log(1/\delta)$.
We derive our sample bound by computing the fat-shattering dimension of the setting and showing that it is linear in the maximum patience of a buyer, i.e., $\wmax$.
We also consider an online setting and bound the regret  with respect to the optimal pure strategy by
$\widetilde{\cO}\paren{\sqrt{T}\sqrt{\wmax}}$.

For mixed strategies, we characterize the buyer's best response strategy as a threshold strategy, and show that the expected buyer partial utility decreases with the time steps. We give an algorithm to compute the optimal mixed selling strategy which is exponential in the maximum patience and polynomial in the support of the distribution.
 We give a general sample complexity bound for the approximate optimal mixed strategy which is polynomial in $\wmax,1/\eps,\log(1/\delta)$.
 We also consider an online setting and bound the regret  with respect to the optimal mixed strategy by $\widetilde{\cO}\paren{ T^{2/3}\wmax^{4/3}}$.

Our work leaves many interesting open problems. 
\begin{itemize}[leftmargin=0.5cm]
    \item \textbf{Computational.}
It is unclear whether one can compute the optimal mixed selling strategy in time polynomial in $\wmax$, or maybe there is a hardness result as in other Stackelberg games with large action spaces \cite{blum2019computing}.
\item \textbf{Sample complexity.} There is a gap between our upper and lower bounds for the sample complexity both for the pure and mixed strategies. Resolving those gaps would be highly interesting.
\item \textbf{Online learning.} First, it is unclear whether one can get a regret bound of $\sqrt{T}$ with respect to the optimal mixed strategy, or whether there is a lower bound of $T^{2/3}$. Another interesting challenge for the online learning is to consider more limited feedback models, for example, when we observes only whether a purchase was made.
\item \textbf{Interaction between buyers.}
In our model the seller interacts with each buyer separately. It would be interesting to consider a model where there are potentially multiple buyers per interaction. At each step there are some buyers that arrive and other buyers that did not purchase at the previous step, and can still benefit from buying. Such a model will introduce many new and intriguing research challenges.
\end{itemize}

\newpage

\section*{Acknowledgments}
We deeply thank Aryeh Kontorovich and Yuval Dagan for very helpful discussions.
This project has received funding from the European Research Council (ERC) under the European Union’s Horizon 2020 research and innovation program (grant agreement No. 882396), by the Israel Science Foundation (grants 993/17, 1602/19), Tel Aviv University Center for AI and Data Science (TAD), and the Yandex Initiative for Machine Learning at Tel Aviv University.
I.A. is supported by the Vatat Scholarship from the Israeli Council for Higher Education and by Kreitman school of Advanced Graduate Studies.

\bibliography{paperbib}

\appendix

\section{Proofs for \cref{sec:characterization-and-planning}}\label{app:characterization-and-planning} 

\subsection{Warm-up: buyer's value and patience are independent} \label{app:charact-warmup}

In this section, we will show that whenever the buyer's value and patience are independent under the joint distribution $\D$, the optimal selling strategy is a fixed price. To show that, we will define a reduction from a Myerson mechanism to our setting, and use the Myerson mechanism's optimality to deduce the optimality of our mechanism.

By the revelation principle \cite{myerson1979incentive}, any selling mechanism can be described by an \textit{incentive-compatible direct selling mechanism}, in which the buyers are supposed to honestly reveal their value.
Hence, in order to find the optimal mechanism, it suffices to consider only this simple class of mechanisms.

\begin{theorem}\cite{myerson1981optimal}\label{thm:myerson}
Assume the distribution over buyers' values is regular. Then, the incentive-compatible direct selling mechanism for a single item that maximizes the seller's expected revenue is the sealed-bid second-price auction with a reserve price $\ppst$ satisfying $\ppst h(\ppst)=1$.
\end{theorem}

In our setting, there is only one buyer per interaction, and thus the Myerson's mechanism degenerates into the mechanism according to which the seller posts the reserve price $\ppst$. Denote the marginal distributions of the buyer's value and patience by $\D_v$ and $\D_w$, respectively.

\begin{theorem}\label{thm:independent-optimal-is-fixed}
If the buyer's value and patience are independent under the joint distribution $\D$, and the marginal distribution $\D_v$ of the buyer's value is regular, then the optimal selling strategy is a fixed price strategy $\pst = (\ppst, \ldots, \ppst)$, satisfying $\ppst h(\ppst)=1$.
\end{theorem}

The crux of the proof of the theorem is a reduction in which given a selling strategy $\Pd$ and its corresponding selling mechanism, we build an incentive-compatible direct selling mechanism for a buyer whose type includes only a value. Upon receiving a value $v$ from the buyer, we sample a patience $w$ according to $\D_w$, execute the mechanism on the best-response strategy for buyer type $(v,w)$, and yield its output. Due to the optimality of the Myerson's mechanism, the revenue of our mechanism may not exceed the Myerson's revenue.

 When the buyer's value and patience are dependent, the mechanism we build during the above reduction is not necessarily incentive-compatible, as can be seen in \cref{rmk:dependent}.

\subsubsection{Proof of \cref{thm:independent-optimal-is-fixed}}

To prove \cref{thm:independent-optimal-is-fixed}, we use the following definitions.

\paragraph{Definitions.}
Our model can be described as a multi-stage \textit{selling mechanism}. A selling mechanism $M_\Pd$ of a selling strategy $\Pd$, receives a best-response strategy $\pistP$ of a buyer type $(v,w)$ against $\Pd$, and by simulating the game on the profile $(\Pd, \pistP)$, outputs the probability $q_{v,w}$ of the buyer getting the item, and his payment $\mus$ to the seller.

A selling mechanism in our model, as defined above, is related to a seemingly broader concept. In general, a \textit{single item selling mechanism} is defined in relation to a set of $n$ buyers, where each buyer $j\in[n]$ is associated with a value $v_j$, indicating how much he might be willing to pay for the item, and a set $\Zc_j$ of signals that he can send to the seller. Given a signal vector $\zb = (z_1, \ldots, z_n) \in \Zc_1\times\ldots\times\Zc_n$, the mechanism outputs for each buyer $j\in[n]$, the probability $q_{\zb}^j$ of the buyer getting the item and his payment $\mu_{\zb}^j$ to the seller. The seller would like to find a mechanism maximizing her excepted revenue, while she does not know the buyers' values. By the revelation principle \cite{myerson1979incentive}, in order to find the optimal mechanism, it suffices to consider only \textit{incentive-compatible direct selling mechanism}, in which the buyers are supposed to honestly reveal their value.

\begin{proof}[of \cref{thm:independent-optimal-is-fixed}]
First, we show that any selling strategy would give the seller a total revenue of at most the Myerson revenue corresponding to $\D_v$. Assume by contradiction that there exists a selling strategy $\Pd$, such that the mechanism $\MP$ yields the seller a higher total revenue than the Myerson revenue. Let $\MDv$ be the following direct selling mechanism: 
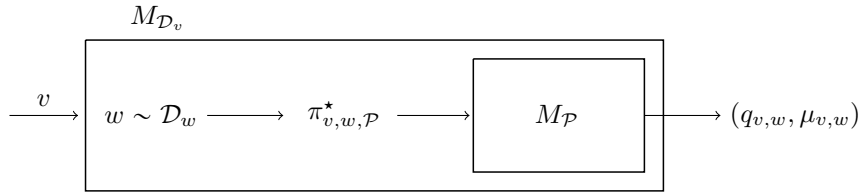
\begin{center}
\begin{tikzpicture}
\draw  [line   width=0.2mm](-1.6,1)--node[above, very near start]{ $\MDv$}(6,1)--(6,-1)--(-1.6,-1)--(-1.6,1);
\draw  [line   width=0.2mm](3.5,0.75)--(5.75,0.75)--(5.75,-0.75)--(3.5,-0.75)--(3.5,0.75);
\draw [->](-2.6,0)--node[above]{ $v$}(-1.7,0);
\draw [->](2.5,0)--(3.45,0);
\draw [->](5.75,0)--(6.75,0) node[right]{ $\left(q_{v,w},\mu_{v,w}\right)$};
\node at (1.8,0) { $\pistP$};
\node at (4.6,0) { $\MP$};
\draw [->](0,0)node[left]{ $w\sim \D_w$}--(1,0) ;
\end{tikzpicture}
\end{center}

The buyer reports a value $v$ (not necessarily his true one) to the mechanism $\MDv$. Then, $\MDv$ selects a patience $w\sim \D_w$, and plays the best-response strategy $\pistP$ of buyer type $(v,w)$ against selling strategy $\Pd$ in the mechanism $\MP$. The buyer then gets the item with probability $\qs$ and pays $\mus$ if he gets the item, which insures that the mechanism satisfies individually rationality. Hence, the expected utility of the buyer whose value is $v$ is $\Eb{w \sim \D_w}{\left(v-\mu_{v,w}\right)\cdot q_{v,w}}$.

Since $\pistP$ is a best-response strategy and the buyer's value and patience are independent under $\D$, the truth-telling is a dominant strategy in $\MDv$. Indeed, the expected utility of the buyer whose value is $v$ and reports $v'$ is given by $\Eb{w \sim \D_w}{\left(v-\mu_{v',w}\right)\cdot q_{v',w}}$ which satisfies the following:
\begin{align}
\Eu{w \sim \D_w}{\left(v-\mu_{v',w}\right)\cdot q_{v',w}}
&=\sum_{w=1}^{\wmax} \D[w|v'] \cdot \left(v-\mu_{v',w}\right)\cdot q_{v',w}\label{eq:total}\\
&=\sum_{w=1}^{\wmax} \D[w] \cdot \left(v-\mu_{v',w}\right)\cdot q_{v',w}\label{eq:independent}\\
&\leq\sum_{w=1}^{\wmax} \D[w]\cdot\left(v-\mu_{v,w}\right) \cdot q_{v,w}\label{ineq:equil}\\
&=\Eu{w \sim \D_w}{\left(v-\mu_{v,w}\right)\cdot q_{v,w}}\label{eq:v},
\end{align}
where \cref{eq:independent} holds due to independence, Inequality (\ref{ineq:equil}) holds due to the optimality of the buying strategy, and \cref{eq:v} follows from \cref{eq:total,eq:independent} with $v' = v$. In particular, $\MDv$ is incentive-compatible.

Furthermore, due to independence, the expected revenue of the seller in $\MP$ is
\[
\Eu{(v,w)\sim\D}{\mus \cdot \qs}=\Eu{v\sim\D_v}{\Eu{w \sim \D_w}{\mus \cdot \qs}}.
\]
That is, the expected revenue  of the seller in  $\MP$ equals to that in the mechanism $\MDv$. In particular, $\MDv$ is an incentive-compatible direct selling mechanism which yields the seller a higher revenue than the Myerson revenue corresponding to $\D_v$, in contradiction to the optimality of Myerson mechanism (\cref{thm:myerson}).

To complete the proof, we show that the expected revenue of a fixed pricing $\pst = (\ppst, \ldots, \ppst)$, satisfying $\ppst h(\ppst)=1$, is the Myerson revenue. Indeed, following the regularity assumption, by \cref{thm:myerson}, the sealed-bid second-price auction with the reserved price $\ppst$ is an optimal selling mechanism. Since in our setting, there is only one buyer per interaction, this auction is is simply the auction according to which the seller posts the price $\ppst$, and a buyer with value $v$ gets the item if $v \geq \ppst$. That is, the Myerson mechanism in this case is $M_{\pst}$, so the total revenue of $\pst$ is the Myerson revenue, as required. 
\end{proof}

\begin{remark}\label{rmk:dependent}
The above proof does not necessarily hold when the buyer's value and patience are dependent. Consider the following joint distribution $\D$:
\[
v \sim U[0,1], \quad 
w|v = 
\begin{cases}
1, & v \in [\frac{1}{2}, 1]\\
2, & v \in [0,\frac{1}{2}).
\end{cases}
\]
Note that the marginal distribution $\D_v$ of the buyer’s value is regular. However, against the pure selling strategy $\p = (3/4, 1/4)$, the best-response buying strategy of a buyer type $(v,w)$ is
\[
\pi^{\star}_{v,w,\p} =
\begin{cases}
(1,0), & v \in [\frac{3}{4},1] \land w=1 \\
(0,1), & v \in [\frac{1}{4}, \frac{1}{2}) \land w=2\\
(0,0), & \text{otherwise}.
\end{cases}
\]
Hence, the output functions of the reduced mechanism $\MDv$ are
\[
q_{v,w} = 
\begin{cases}
1, & v \in [\frac{3}{4},1] \land w=1 \\
1, & v \in [\frac{1}{4}, \frac{1}{2}) \land w=2\\
0, & \text{otherwise}
\end{cases}
, \quad \mu_{v,w} = 
\begin{cases}
\frac{3}{4}, & v \in [\frac{3}{4},1] \land w=1 \\
\frac{1}{4}, & v \in [\frac{1}{4}, \frac{1}{2}) \land w=2\\
0, & \text{otherwise}.
\end{cases}
\]
Therefore, a buyer whose value is $v \geq \frac{1}{2}$ is better of reporting $v' \in [\frac{1}{4}, \frac{1}{2})$ in $\MDv$, so the mechanism is no longer incentive-compatible.
\end{remark}

\subsection{Separating optimal fixed price, pure strategy and mixed strategy}\label{app:charact-hierarchy}
In this section, we show a separation between the optimal fixed price, optimal pure strategy and optimal mixed strategy.
This in particular implies that the optimal selling strategy is not necessarily a fixed price. 
In addition, it implies that in general the optimal strategy needs to be a mixed strategy rather than a pure strategy.

\begin{theorem}\label{thm:selling-strategy-hierarchy}
There exists a distribution $\D_1$ under which any optimal selling strategy is pure but not a fixed price, i.e., there exists a pure strategy $\p$ such that for any fixed pricing $\p'=(p', \ldots,p')$ we have $\rev(\p)>\rev(\p')$. 
In addition, there exists a distribution $\D_2$ under which any optimal selling strategy is mixed but not pure, i.e., there exists a mixed strategy $\Pd$ such that for any pure strategy $\p$ we have $\rev(\Pd)>\rev(\p)$.
\end{theorem}

To separate the pure strategies from fixed prices, we consider a distribution $\D_1$ which is uniform over the buyer types
$\{(1/3,3),(2/3,2),(1,1)\}$.
We show that the optimal fixed price strategy is $2/3$ with excepted revenue of
$4/9$, whereas the optimal pure selling strategy is $(1,\frac{2}{3},\frac{1}{3})$ with an excepted revenue of $2/3 = 6/9 > 4/9$.

To separate the mixed strategies from pure strategies, we consider a distribution $\D_2$ which is uniform over the buyer types
$\{(1/3,2),(2/3,1),(1,2)\}$.
We show that the optimal pure strategy has an excepted revenue of $4/9$ and there exists a mixed selling strategy with expected revenue of $13/27 > 12/27 =4/9$. 

\begin{proof}[of \cref{thm:selling-strategy-hierarchy}]
First, we present an example separating the pure strategies from fixed prices. Assume $\wmax= 3$. Consider the following joint distribution $\D_1$:
\[
\D_1[v=1/3, w=3] =\D_1[v=2/3, w=2] = \D_1[{v=1, w=1}] = \frac{1}{3}
\]
Then, the pure selling strategy $(1,\frac{2}{3},\frac{1}{3})$ is optimal since its revenue is the buyer's expected value. In particular, it gives an excepted revenue of 
\[
\frac{1}{3}\cdot\frac{1}{3}+\frac{2}{3}\cdot\frac{1}{3}+1\cdot\frac{1}{3}=\frac{2}{3},
\]
whereas the best fixed price strategy is $(\frac{2}{3},\frac{2}{3},\frac{2}{3})$ and has an excepted revenue of $\frac{2}{3} \cdot \frac{2}{3} = \frac{4}{9} < \frac{6}{9}=\frac{2}{3}$.

Second, we present an example separating the mixed strategies from pure strategies.
Assume $\wmax = 2$. Consider the following joint distribution $\D_2$:
\[
\D_2[v=1/3, w=2] = \D_2[v=2/3, w=1] = \D_2[{v=1, w=2}] = \frac{1}{3}.
\]
Then, the excepted revenue for each pure selling strategy is
\begin{center}
\begin{tabular}{ |c||c|c|c|c|c|c|c|c|c| } 
 \hline
 Strategy & $(1,1)$ & $(1,\frac{2}{3})$ & $(1,\frac{1}{3})$ & $(\frac{2}{3},1)$ & $(\frac{2}{3},\frac{2}{3})$ & $(\frac{2}{3},\frac{1}{3})$ & $(\frac{1}{3},1)$ &
 $(\frac{1}{3},\frac{2}{3})$ &
 $(\frac{1}{3},\frac{1}{3})$\\ 
 \hline\hline
 Revenue & $\frac{1}{3}$ & $\frac{2}{9}$ & $\frac{2}{9}$ & $\frac{4}{9}$& $\frac{4}{9}$ & $\frac{4}{9}$ & $\frac{1}{3}$ &
 $\frac{1}{3}$ & $\frac{1}{3}$
 \\ 
 \hline
\end{tabular}
\end{center}
We obtain that the optimal pure strategies are $(\frac{2}{3}, 1), (\frac{2}{3},\frac{2}{3})$ and $(\frac{2}{3},\frac{1}{3})$ which all give an excepted revenue of $4/9$. Consider the following mixed strategy $\Pd$:
\[
\Pd[(2/3,1/3)] = \frac{1}{3}, \quad \Pd[(2/3,1)] = \frac{2}{3}
\]

A buyer with value $2/3$ (and thus with patience $1$) would buy at the first step at price $2/3$. A buyer whose value is $1/3$ would wait for the second step, since he cannot afford to buy at the first step, and would buy with probability $1/3$. 

Consider a buyer with value $1$. If he buys at the first step, then his utility is $1-\frac{2}{3}=\frac{1}{3}$. However, if he waits for the second step, his utility is $\frac{1}{3}(1-\frac{1}{3}) = \frac{2}{9} < \frac{3}{9}=\frac{1}{3}$. Hence, he would also buy at the first step. Therefore, the total revenue of the seller from mixed strategy $\Pd$ is
\[
\rev(\Pd) = \frac{2}{3} \cdot \frac{1}{3} + \frac{1}{3} \cdot \frac{1}{3}\cdot\frac{1}{3} +  \frac{2}{3} \cdot \frac{1}{3}= \frac{13}{27} > \frac{12}{27} =\frac{4}{9}.
\]
That is, the mixed selling strategy $\Pd$ is preferable to any optimal pure strategy.
\end{proof}

\subsection{Optimal pure selling strategy} \label{app:charact-and-plan-pure}

\subsubsection{Proof of \cref{thm:optimal-pure}}\label{app:optimal-pure}
Before we prove \cref{thm:optimal-pure}, we present several lemmas: 
 
\begin{lemma}\label{lem:non-increasing-pure}
For every pure selling strategy $\p$, there exists a non-increasing pure strategy $\p'$ such that $\rev(\p') = \rev(\p)$. Moreover, in response to $\p'$, each buyer of type $(v,w)$ buys exactly at step $w$, as long as $v \geq p'_w$ (otherwise, he does not buy at all).
\end{lemma}
\begin{proof}
Let $\p$ be a pure strategy. If $\p$ is increasing, then there is a first step $j$ in which the price is increasing, i.e., $p_j > p_{j-1}$. 
A strategic buyer would not buy at this step $j$ since he can buy at step $j-1$ at a lower price. This implies that no buyer buys at step $j$.

Consider a modification of $\p$ to $\tilde{\p}$ such that the price at step $j$ is $\tilde{p}_j=p_{j-1}$. Clearly, we have $\rev(\p)=\rev(\tilde{\p})$, since the same buyer types buy and they pay the same prices.

If we repeat this process and remove any increase in prices in this way, we finally get a non-increasing pricing $\p'$ with the same revenue as that of $\p$.

Now, consider the best-response of a buyer type $(v,w)$ against the non-increasing pricing $\p'$. By monotonicity, the lowest price in $\p'$ such that $i \leq w$ is $p'_w$. Hence, the buyer is better of waiting to step $w$, and buy if he is able to, i.e. $v \geq p'_w$. 
\end{proof}

Now, we give a general characterization for the support of mixed selling strategies. While this is straightforward in the case of pure strategies, it requires work in the mixed case.

\begin{lemma}\label{lem:mixed-over-value-range}
Assume the support of the marginal distribution of the buyer’s value is contained in $[\vmin, \vmax] \subseteq [0,1]$. Then, for every mixed selling strategy $\Pd$, there exists a mixed strategy $\Pd'$ suing only prices from $[\vmin, \vmax]$, such that $\rev(\Pd') \geq \rev(\Pd)$.
\end{lemma}

\begin{proof}
Let $\Pd$ be a mixed selling strategy. We perform two modifications to the given strategy: the first will ensure that the maximum offered price is $\vmax$ whereas the second will ensure that the minimum offered price is $\vmin$. Each of them will give us a strategy with a total revenue of at least as much as that of $\Pd$. Thus the strategy obtained from making both changes is the desired one. 

Define the following transformation $f: \R^{\wmax} \rightarrow \R^{\wmax}$ over the set of pure strategies:
\[
[f(\p)]_i =
\begin{cases}
p_i, & p_i \leq \vmax\\
\vmax, & p_i > \vmax
\end{cases}
\]
That is, we lower any price in $\p$ above $ \vmax$  to $\vmax$. Denote by $\Pd'$ the mixed strategy obtained from invoking $f$ on $\p \sim \Pd$, i.e., for any $\z=f(\p)$ let $Z=\{\p:f(\p)=\z\}$ and we set $\Pd'(\z)=\Pd(Z)$. We show that $\rev(\Pd') \geq \rev(\Pd)$. First, since we only lower prices, the probability of a sale may only increase. Second, under the original strategy, no buyer could afford to buy at a price above $ \vmax$. So such prices do not contribute to the expected revenue of $\Pd$. This implies that $\rev(\Pd')\geq\rev(\Pd)$.

Now we handle the minimum value $\vmin$. We define the following transformation $g: \R^{\wmax} \rightarrow \R^{\wmax}$ over the set of pure strategies:
\[
[g(\p)]_i =
\begin{cases}
p_i, & p_i \geq \vmin\\
\vmin, & p_i < \vmin
\end{cases}
\]
That is, we raise any price $< \vmin$ in $\p$ to $\vmin$. Denote by $\Pd''$ the mixed strategy obtained from invoking $g$ on $\p \sim \Pd$, i.e., for any $\mathbf{z}=g(\p)$ let $Z=\{\p:g(\p)=\mathbf{\z}\}$ and we set $\Pd''(\z)=\Pd(Z)$. To show that $\rev(\Pd'') \geq \rev(\Pd)$, by \cref{lem:sufficient-conds-for-greater-rev}, it suffices to prove that for any buyer type $(v,w)$ after any history $\p_{1:i}$ of length $i\in [\wmax]$, both of the following conditions are met: (1) if the buyer preferred to buy at step $i$, he still buys at this step; (2) his utility from step $i$ onwards does not increase. We prove it by backward induction on the history length $i$. For $i=\wmax$, since this is the final step, the buyer's decision is based only on his value, and he obviously can still buy if the price was increased to $\vmin$. In this case, his utility may only decrease. 
Assume the inductive hypothesis holds for any $j > i$ and prove for step $i$. Let $(v,w)$ be a buyer type with patience $w \geq i$. Consider the following cases:
\begin{itemize}
    \item The price was not changed at this step, i.e.,  $p''_i=p_i$, and the buyer preferred to buy at step $i$. By the inductive hypothesis (condition 2), his future utility from step $i+1$ onward was not increased. Hence, he buys at step $i$, and thus his utility from step $i$ onwards also remains the same in $\Pd$ and $\Pd''$.
    
    \item The price was changed to $\vmin$, i.e.,  $p''_i=\vmin>p_i$,. Since the buyer is offered the minimum price, $\vmin$, and any future price in $\Pd''$ is at least the minimum price $\vmin$, the buyer would necessarily buy at step $i$. As the price was previously lower, the buyer’s utility from this step onwards may only decrease. 
\end{itemize}
Therefore, both conditions are satisfied as required and the lemma foolows.
\end{proof}

In the last lemma, we used the following lemma:
\begin{lemma}\label{lem:sufficient-conds-for-greater-rev}
Let $\Pd, \Pd'$ be mixed selling strategies. If for any buyer type $(v,w)$ after any history $p_{1:i}=(p_1, \ldots, p_i)$ of length $i \in [\wmax]$, it holds that:
\begin{enumerate}
    \item If the buyer buys at step $i$ under $\Pd$, he buys also at step $i$ under $\Pd'$, i.e., if
    $\piP^i(\p_{1:i}) = \e_i$, then  $\pi_{v,w,\Pd'}^i(\p_{1:i}) = \e_i$.
    
    \item  The buyer's partial utility from step $i$ onwards under $\Pd'$ is greater than his utility under $\Pd$ by at most constant $c_i$, i.e.,
    $u_{v,w,\Pd'}^i(\pi^{\star}_{v,w,\Pd'}; \p_{1:i-1}) \leq \utilP^i(\pistP; \p_{1:i-1}) + c_i $,
\end{enumerate}
where $\pistP$ is a best-response buying strategy of a buyer type $(v,w)$ against $\Pd$. Then,
\[
r(\Pd') \geq r(\Pd) -  c_1.
\]
\end{lemma}

\begin{proof}
Let $q_{v,w}$ and $q'_{v,w}$ be the probability of a buyer of type $(v,w)$
getting the item under selling strategies $\Pd$ and $\Pd'$, respectively.
Similarly, let $\mu_{v,w}$ and $\mu'_{v,w}$ be the excepted payment of a buyer of type $(v,w)$, given the item was sold, under selling strategies $\Pd$ and $\Pd'$, respectively. Then, by definition, the total revenue of selling strategy $\Pd$ is $\rev(\Pd) = \Eb{(v,w) \sim \D}{\mu_{v,w}q_{v,w}}$, whereas the utility of a best-response buying strategy $\piP$ for a buyer type $(v,w)$ is $\utilP(\pistP) = (v-\mu_{v,w})q_{v,w}$. Similarly, $\rev(\Pd') = \Eb{(v,w) \sim \D}{\mu'_{v,w}q'_{v,w}}$ and $\util_{v,w,\Pd'}(\pi^{\star}_{v,w,\Pd'}) = (v-\mu'_{v,w})q'_{v,w}$
for any buyer type $(v,w)$.

From condition (1), we obtain that the sale probability does not decrease, i.e., $q'_{v,w} \geq q_{v,w}$ for any buyer type $(v,w)$. Hence, by condition (2) applied to $i=1$, we get that 
\[
(v - \mu_{v,w})q_{v,w} +c_1 \geq (v - \mu'_{v,w})q'_{v,w} \geq (v - \mu'_{v,w})q_{v,w}.
\]
Thus,
\[
\mu_{v,w}q_{v,w} \leq \mu'_{v,w} q_{v,w} + c_1 \leq \mu'_{v,w} q'_{v,w} +c_1.
\]
From which it follows that 
\[
\rev(\Pd) = \Eu{(v,w) \sim \D}{\mu_{v,w}q_{v,w}} \leq \Eu{(v,w) \sim \D}{\mu'_{v,w}q'_{v,w}} +c_1 = \rev(\Pd') +c_1,
\]
as required.
\end{proof}

\begin{lemma} \label{lem:pure-over-values}
Assume the support of the marginal distribution of the buyer’s value, $\V$, is a finite subset of $[0,1]$. Then, for every pure selling strategy $\p$, there exists a pure strategy $\p'$ using only prices from $\V$, such that $\rev(\p') \geq \rev(\p)$.
\end{lemma}

\begin{proof}
Let $\p$ be a pure pricing over $[0,1]$. By \cref{lem:non-increasing-pure}, we can assume it is non-increasing. Increase each price up to the nearest value in $\V$, to obtain a pricing $\p'$. Note that $\p'$ is also a non-increasing pricing. Thus, we can assume that under both strategies,
the buyer buys at his final step, if he is able. Following our modification, if the buyer could afford to buy the item previously, then he can buy now as well. Therefore, the change can only benefit the seller.
\end{proof}

Combining \cref{lem:non-increasing-pure,lem:mixed-over-value-range,lem:pure-over-values}, we conclude:
\thmoptimalpure*

\subsection{Computing an optimal pure selling strategy}\label{app:pure-planning}

In this section, we compute efficiently an optimal pure selling strategy.

\paragraph{Overview of the algorithm.} In Algorithm \ref{algo:pure}, we present a dynamic programming algorithm to compute an optimal pure selling strategy. By \cref{thm:optimal-pure}, there exists a pure optimal selling strategy with non-increasing prices. Hence, the algorithm generates a pricing with non-increasing prices. By \cref{lem:non-increasing-pure}, for non-increasing pricing, the buyer buys at the last step in his patience window, as long as the price is lower than his value. The algorithm takes advantage of this in order to simplify the calculation of the seller's revenue.

The algorithm works as follows. It starts at step $\wmax$ and goes backwards in time steps. For each step $i$, it computes for each price $p_i$, the maximum revenue $r_i(p_i)$ from step $i$ onwards of the best non-increasing sequence of prices from step $i$, where the price at step $i$ is $p_i$.
Assume that we already computed the optimal revenue $r_{i+1}(p_{i+1})$ for step $i+1$ and any price $p_{i+1}$, we now compute $r_i(p_i)$ for step $i$ and price $p_i$. We set $r_i(p_i)$ to be the sum of two terms. The first term is $p_i$ times the probability that a buyer has patience $i$ and value at least $p_i$. The second term is maximum over $p_{i+1}\leq p_i$ of $r_{i+1}(p_{i+1})$. The first term is the excepted revenue we get from buyer's types that buy at step $i$ (and in particular have patience $i$). The second term is the future revenue we get for steps $i+1$ to $\wmax$, under the assumption that the price sequence is non-increasing. Hence, we have established the following theorem.

\begin{algorithm}
\caption{Optimal \textbf{pure} selling strategy}\label{algo:pure}
\textbf{Input:} Distribution $\D$ over $\V \times [\wmax]$, where $\V\subseteq[0,1]$ is a finite set.
\vspace{2mm}
\\
\textbf{Declare:} for each step $i$, let partial revenue $r_i: \V \rightarrow [0,1]$ and price $a_{i+1}: \V \rightarrow \V$.
\\
\textbf{Initialize:} for every price $p \in \V$, set $r_{\wmax+1}(p) \leftarrow 0$.
\\
\begin{enumerate}[leftmargin=18pt,rightmargin=10pt,itemsep=1pt,topsep=3pt]
    \item For each step $i \leftarrow \wmax, \ldots, 1$, and price $p_i\in\V$: \label{pline:compute-revs}
    \begin{enumerate}
       \item $r_i(p_i) \leftarrow \Pr_{(v,w) \sim \D}\left(v \geq p_i, w=i\right) \cdot p_i + \max\set{r_{i+1}(p_{i+1}): p_i\geq p_{i+1}\in \V}$ 
        \item $a_{i+1}(p_i) \leftarrow \argmax\set{r_{i+1}(p_{i+1}) : p_i\geq p_{i+1}\in \V}$
    \end{enumerate}
    \vspace{1mm}
    \item[{\color{gray} \#}] {\color{gray} Reconstruct the optimal pricing $\pst$.}
    \item $\ppst_1 \leftarrow \argmax \set{r_1(p_1) : p_1 \in \V}$ \label{pline:set-p-one}
    \item For each step $i \leftarrow 1, \ldots, \wmax-1$: \label{pline:set-ps}
    \begin{enumerate}
        \item $\ppst_{i+1} \leftarrow a_{i+1}(\ppst_i)$
    \end{enumerate}
\end{enumerate}
\textbf{Output:} Optimal pure selling strategy $\pst = (\ppst_1, \ldots, \ppst_{\wmax})$.
\end{algorithm}

\thmpureplanning*

The proof of \cref{thm:pure-planning} is a straightforward backward induction to show that at each step $i$, the revenue of $\pst$ from step $i$ onwards is maximal. 

\begin{proof}[of \cref{thm:pure-planning}]
We prove the correctness and running time of \cref{algo:pure}.
\paragraph{Correctness.}
First, we show that for any step $i \in [\wmax]$ and price $p_i \in \V$, $r_i(p_i)$ is the maximum revenue from step $i+1$ onwards when the price at step $i$ is $p_i$. We prove it by backward induction on $i$. Clearly, $r_{\wmax+1}(p) =0$ for any $p \in \V$. Assume the inductive hypothesis holds for any step $j > i$, and prove for $i$. Let $p_i \in \V$ be a price. According to the algorithm, $r_i(p_i) = \Pr_{(v,w) \sim \D}\left(v \geq p_i, w\right) \cdot p_i + \max \set{r_{i+1}(p_{i+1}) : p_i \geq p_{i+1} \in \V}$. Since we build a non-increasing pricing, by \cref{lem:non-increasing-pure}, each buyer type may buy only at his final step, so the excepted revenue at step $i$ when the price is $p_i$ is $\Pr_{(v,w) \sim \D}\left(v \geq p_i, w\right) \cdot p_i$. By induction assumption, $r_{i+1}(p_{i+1})$ is the maximum revenue from step $i+1$ onwards when the price at step $i+1$ is $p_{i+1}$. Thus, $\max \set{r_{i+1}(p_{i+1}) : p_i \geq p_{i+1} \in \V}$ is the maximum revenue from step $i+1$ onwards. Summing these two revenues would obtain the required. Therefore, we get that $\pst = (\ppst_1, \ldots, \ppst_{\wmax})$ is an optimal pure selling strategy.

\paragraph{Running time.}
The initialization part takes $\O(|\V|)$. Maximizing over $\O(|\V|)$ takes $\O(|\V|)$ and thus Line (\ref{pline:compute-revs}) takes $\O(|V|^2\wmax)$. Line (\ref{pline:set-p-one}) takes $\O(|V|)$ and Line (\ref{pline:set-ps}) takes $\O(\wmax)$. In total, the algorithm takes $\O(|\V|^2\wmax)$.
\end{proof}

\subsection*{Handling continuous values}

When the set of buyer's values $\V$ is $[0,1]$, we take the possible prices to be the discretized set $P_\eps=\{0,\eps,2\eps, \ldots, 1\}$, and guarantee an \emph{$\eps$-optimal} selling strategy for $\eps>0$. That is, a selling strategy whose revenue is at most $\eps$ away from the optimal revenue.

The following lemma shows that the discretization gives a good approximation.

\begin{lemma}\label{lem:eps-optimal-pure}
Assume the buyer's value set is $[0,1]$.
Let $\phat$ be an optimal pricing with respect to the pure strategies over a discretization $P_\eps = \{0,\eps,2\eps, \ldots, 1\}$ of $[0,1]$, for any $\eps > 0$. Then, $\phat$ is an $\eps$-optimal with respect to the optimal pure strategies over $[0,1]$.
\end{lemma}
\begin{proof}
By \cref{thm:optimal-pure}, there exists an optimal non-increasing pure strategy $\pst$ over $[0,1]$. Decrease each price down to the nearest $i \eps$, for $0\leq i \leq 1/\eps$, to get a pricing $\p'$. Since $\p'$ is also non-increasing, we can assume that under both strategies, the buyer buys exactly at the step where his patience runs out, if he can afford it.
As we only lower prices, if the buyer made a purchase before the change, he will do so now as well. However, the excepted loss of this modification is at most $\eps$, that is, $\rev(\pst) - \rev(\p') \leq \eps$. Recall that $\phat$ is the optimal pure strategy with respect to pure strategies over $P$. Clearly $\rev(\pst)\geq\rev(\phat)$ and $\rev(\phat)\geq\rev(\p')$. Therefore,
\[
\rev(\pst) - \rev(\phat) \leq \rev(\pst) - \rev(\p') \leq \eps.
\]
\end{proof}

The idea behind the proof of \cref{lem:eps-optimal-pure} is to modify an optimal non-increasing pricing by decreasing each its price down to the nearest price $i\eps\in P_\eps$, so the resulting pricing is still non-increasing, and since any buyer type that could buy previously, can still buy now, the revenue decreases by at most $\eps$.

We conclude the following result.
\begin{theorem}\label{thm:eps-optimal-pure-planning}
There exists an algorithm that for any $\eps>0$ and a distribution $\D$ over $[0,1] \times [\wmax]$, returns an $\eps$-optimal pure strategy, and runs in time $\O(\wmax/\eps^2)$.
\end{theorem}

\begin{proof}
Let $\eps > 0$. We discretize the buyer's value set $\V$ to obtain a finite set of values $\V_\eps$ with $\O(1/\eps)$ values and run \cref{algo:pure}, so the theorem follows from \cref{thm:pure-planning,lem:eps-optimal-pure}.
\end{proof}

\subsection{Optimal mixed selling strategy}
\subsubsection{Buyer characterization}\label{app:optimal-threshold-buying} 
In this section, we focus on the case where the seller strategy is mixed. Our main goal is to characterize the buyer's best response to a seller's mixed strategy.
We present a simple class of buying strategies, which we call \textit{threshold strategies}. We show that for any mixed seller strategy, there exists a buyer's best response strategy which is a threshold strategy. 

\paragraph{Notation.} For any mixed selling strategy $\Pd$, we denote the marginal distribution of prices at step $i$ by $\Pd_i$. The conditional selling strategy, given a history of prices $p_1, \ldots, p_i$ over the set $\set{\p'\in[0,1]^{\wmax} : \p'_{1:i}=\p_{1:i}}$, is denoted by $\Pd|\p_{1:i}$. 

The \emph{partial utility} is the buyer's utility from a step $i$ onwards, and defined formally as follows.
\begin{definition}[Partial utility]
Denote the buyer's strategy $\pi_{v,w}$ for a buyer type $(v,w)$ from step $i\in[\wmax]$ onwards as $\pi_{v,w}(\p; i)=\e_j$  if $j\geq i$ is the first step where $\pi^j_{v,w}(\p_{1:j})=\e_j$ or $\pi_{v,w}(\p; i)=\e_0$ if no such index $j$ exists.
Define the \textit{partial utility} from step $i \in [\wmax]$ onwards of a buying strategy $\pi_{v,w}$ for a buyer type $(v,w)$ against a selling strategy $\Pd$, given historical prices $p_1, \ldots, p_{i-1}$, by 
\[
\utilP^i(\pi_{v,w}; \p_{1:i-1}) = 
\Eu{\p \sim \Pd|\p_{1:i-1}}{\utildec_{v,w}(\p, \pi_{v,w}(\p; i))}.
\]
\end{definition}
We compute the partial utility from step $i$ essentially  ``assuming'' the buyer did not buy before step $i$. 
Notice that if $i > w$, 
we have $\utilP^i(\pi_{v,w}; \p_{1:i-1}) = 0$. We further note that the above definition generalizes the definition of utility for $i=1$. Indeed,
\[
\utilP^1(\pi_{v,w}; \p_{1:0}) 
= \Eu{\p \sim \Pd}{\utildec_{v,w}(\p, \pi_{v,w}(\p))}
= \utilP(\pi_{v,w})
\]

The following proposition defines the buyer's best response strategy, according to which at each step $i$, the buyer would buy at step $i$ if he profits from an immediate purchase no less than his future utility.

\begin{proposition}\label{prop:optimal-buying}
Let $\Pd$ be a selling strategy. Consider a buyer of type $(v,w)$.
Let $\pitilde$ be a buyer strategy that satisfies, for every step $i$,
\[
\pitilde^i(\p_{1:i}) = 
\begin{cases}
    \e_i, & v-p_i \geq \utilP^{i+1}(\pitilde; \p_{1:i})\\
    \text{continue}, &v-p_i < \utilP^{i+1}(\pitilde; \p_{1:i}).
\end{cases}
\]
Then, $\pitilde$ is a best response strategy against selling strategy $\Pd$.
\end{proposition}

\begin{proof}
Assume by contradiction that there exists another buying strategy that achieves a greater utility. Let $\piP$ be a buying strategy for a buyer type $(v,w)$ that achieves a greater utility than $\pitilde$, and has a minimum of histories after which it plays different from $\pitilde$. Let $\pj{j}=(p_1, \ldots,p_j)$ be the longest history such that $\piP(\pj{j}) \neq \pitilde(\pj{j})$. We modify $\piP$ to get strategy $\pibar$ by setting $\pibar(\pj{j}) = \pitilde(\pj{j})$, so for any other history $\p'_{1:i}$, we set $\pibar(\p'_{1:i}) = \piP(\p'_{1:i})$. By definition, $\pitilde$ chooses at step $j$ the action that maximizes his partial utility from step $j$,
\[
\util^j(\pitilde; \pj{j-1}) = \max\set{(v-p_j)\cdot \ind{j \leq w}, \utilP^{j+1}(\pitilde; \pj{j})}
\]
After step $j$, both $\pitilde$ and $\piP$ play the same. Hence,
\[
\util^j(\pitilde; \pj{j-1}) = \max\set{(v-p_j)\cdot \ind{j \leq w}, \utilP^{j+1}(\piP; \pj{j})}
\]
Since $\piP$ plays differently at step $j$, we get that
\[
\util^j(\piP; \pj{j-1}) \leq \util^j(\pitilde; \pj{j-1}) = \util^j(\pibar; \pj{j-1})
\]
Because the two strategies, $\piP$ and $\pibar$, are identical except for step $j$ after history $\pj{j}$, this modification can only improve its total utility. Therefore, if $\piP$ achieves a greater utility than $\pitilde$, then so does strategy $\pibar$, but this is in contradiction to the minimality of $\piP$.
\end{proof}

We now can show that the buyer’s best response is a threshold strategy.

\begin{theorem}\label{thm:threshold-equiv}
Let $\Pd$ be a selling strategy. Consider a buyer of type $(v,w)$.

Let $\pihat$  be a threshold strategy with thresholds:
\[
\thetahat^i (\p_{1:i}) = v - \utilP^{i+1}(\pitilde; \p_{1:i}).
\]

Then, for any step $i\in [\wmax]$ and realized prices $p_1,\ldots,p_i$, it holds that $\pihat^i(\pj{i}) = \pitilde^i(\pj{i})$, where $\pitilde$ is defined in \cref{prop:optimal-buying}.
\end{theorem}

\begin{proof}
Let $i \in [\wmax]$ and assume we observed prices $p_1, \ldots, p_i$. We have:
\[
v-p_i \geq \utilP^{i+1}(\pitilde; \p_{1:i}) \Leftrightarrow p_i \leq v - \utilP^{i+1}(\pitilde; \p_{1:i}) = \thetahat^i (\p_{1:i}) 
\]
Therefore, $\pihat^i(\pj{i}) = \pitilde^i(\pj{i})$, as required.
\end{proof}

Intuitively, the threshold at step $i$ is set to the price which makes the buyer ``indifferent'' between buying at step $i$ and continuing to step $i+1$. If the offered price is lower, the buyer makes the purchase and if the price is higher the buyer waits.
Note that if $i \geq w$, in particular if $i \geq \wmax$, we have $\thetahat^i(\p_{1:i}) = v$. Combining \cref{prop:optimal-buying} and \cref{thm:threshold-equiv}, we conclude:

\thmoptimalthresholdbuying*

\begin{remark}
The threshold $\thetahat^i(\pj{i})$ at step $i$ depends on the history $\pj{i}$ through the conditional selling strategy $\Pd|\pj{i}$. The future prices influence the threshold at step $i$ through the partial utility from step $i+1$. The thresholds are set to prices which make the buyer ``indifferent'' between buying at step $i$ and buying in the future. Moreover, the dependence of the thresholds $\thetahat^i(\pj{i})$ on the history $\pj{i}$ is unavoidable. 
\end{remark}

Next, we show that the best-response buying threshold strategy is monotonic in the following sense:

\begin{theorem}\label{thm:buyer-monotone-in-types}
Let $\Pd$ be a mixed selling strategy. Let $(v,w)$ and $(v',w')$ be buyer types such that $v' \geq v$ and $w' \leq w$ and let $\thetahattag(\p_{1:i})$  and $\thetahat(\p_{1:i})$ be their optimal threshold strategies, respectively. Then, for any step $i \in [\wmax]$ and any potentially realized prices $p_1, \ldots, p_i$, it holds that
\[
\thetahattag^i(\p_{1:i})  \geq \thetahat^i(\p_{1:i})
\]
\end{theorem}

\begin{proof}
We prove by backward induction on step $i \in [\wmax]$. For the base of the induction consider $i\geq w'$. Note that $i$ is either equal to the patience of buyer type $(v',w')$ or after it. We have for any realized prices $\p_{1:i}=(p_1, \ldots, p_i)$:
\[
\thetahattag^i(\pj{i}) = v' \geq v \geq v - \utilP^{i+1}(\pihat; \pj{i}) = \thetahat^i(\pj{i})
\]
We now show the inductive hypothesis for $i<w'$.
Assume that the inductive hypothesis holds for any $j > i$ and prove for step $i$. Consider observing prices $\p_{1:i}=(p_1, \ldots, p_i)$. For the sake of simplicity, we omit the notations of $\Pd, \pihat, \pj{j}$ from the thresholds and utilities. That is,
$\thetahats^j = \thetahat^j(\pj{j})$ and $ \utils^{j} = \utilP^{j}(\pihat; \pj{j-1})$.

By the induction hypothesis, we have: $\thetahattags^{i+1} \geq \thetahats^{i+1}$. By the law of total expectation, we have:
\begin{align*}
    \thetahattags^i &= v' - \utiltags^{i+1}\\
    &= \Prob{p_{i+1} \leq \thetahats^{i+1}} \E{v'-(v' - p_{i+1}) \mid p_{i+1} \leq \thetahats^{i+1}}\\ 
    &\quad + \Prob{\thetahats^{i+1} < p_{i+1} \leq \thetahattags^{i+1}} \E{v'-(v' - p_{i+1}) \mid \thetahats^{i+1} < p_{i+1} \leq \thetahattags^{i+1}}\\
    &\quad + \Prob{p_{i+1} > \thetahattags^{i+1}} \E{v' - \utiltags^{i+2} \mid p_{i+1} > \thetahattags^{i+1}}\\
    &= \Prob{p_{i+1} \leq \thetahats^{i+1}} \E{p_{i+1} \mid p_{i+1} \leq \thetahats^{i+1}}\\
    &\quad + \Prob{\thetahats^{i+1} < p_{i+1} \leq \thetahattags^{i+1}} \E{p_{i+1} \mid \thetahats^{i+1} < p_{i+1} \leq \thetahattags^{i+1}}\\
    &\quad + \Prob{p_{i+1} > \thetahattags^{i+1}} \E{\thetahattags^{i+1} \mid p_{i+1} > \thetahattags^{i+1}}\\
    &\geq \Prob{p_{i+1} \leq \thetahats^{i+1}} \E{p_{i+1} \mid p_{i+1} \leq \thetahats^{i+1}}\\
    &\quad + \Prob{\thetahats^{i+1} < p_{i+1} \leq \thetahattags^{i+1}} \E{\thetahats^{i+1} \mid \thetahats^{i+1} < p_{i+1} \leq \thetahattags^{i+1}}\\
    &\quad + \Prob{p_{i+1} > \thetahattags^{i+1}} \E{\thetahats^{i+1} \mid p_{i+1} > \thetahattags^{i+1}}\\
    &= \Prob{p_{i+1} \leq \thetahats^{i+1}} \E{p_{i+1} \mid p_{i+1} \leq \thetahats^{i+1}} + \Prob{p_{i+1} > \thetahats^{i+1}} \E{\thetahats^{i+1} \mid p_{i+1} > \thetahats^{i+1}}\\
    &= \Prob{p_{i+1} \leq \thetahats^{i+1}} \E{v-(v-p_{i+1}) \mid p_{i+1} \leq \thetahats^{i+1}} + \Prob{p_{i+1} > \thetahats^{i+1}} \E{v-\utils^{i+2} \mid p_{i+1} > \thetahats^{i+1}}\\
    &= v- \utils^{i+1}\\
    &= \thetahats^i,
\end{align*}
where the probabilities and the expectations here are over $p_{i+1} \sim \Pd_{i+1}|\p_{1:i}$. The first and third equalities follows from the definition of strategy $\thetahat$, and the inequality follows from the induction assumption, according to which $\thetahattags^{i+1} \geq \thetahats^{i+1}$. 
\end{proof}

This establishes the following corollary, showing that using the threshold strategy if a buyer with a higher valuation does not buy then the buyer with the lower valuation also does not buy. This is obvious for a pure strategy, but for a mixed strategy it is much less obvious.

\begin{corollary} \label{crl:buyer-monotonicity}
Assume that we have two buyers of type $(v,w)$ and $(v',w)$ such that $v'\geq v$. Then,
in any realization where $(v',w)$ did not buy then also buyer type $(v,w)$ does not buy.
\end{corollary}

We consider how the buyer's expected utility changes over time. The following lemma shows that the partial utility of the buyer is non-increasing in the time steps.

\begin{theorem}\label{thm:buyer-monotone-in-time}
Let $\Pd$ be an optimal selling strategy, and $(v,w)$ be a buyer type. Then, for any step $i \in [\wmax]$ and realized prices $p_1, \ldots, p_{i-1}$, it holds that
\[
\Eu{p_i \sim \Pd_{i}|\p_{1:i-1}}{\util^{i+1}(\hat{\pi}; \p_{1:i-1},p_i)} 
\leq \util^i(\hat{\pi}; \p_{1:i-1}),
\]
where $\util^{j}(\hat{\pi}; \p_{1:j-1}) = \utilP^{j}(\pihat; \p_{1:j-1})$ and $\hat{\pi} = \pihat$ is an optimal threshold strategy for buyer type $(v,w)$.
\end{theorem}

\begin{proof}
Let $i \in [\wmax]$ and assume we observed prices $p_1, \ldots, p_{i-1}$. For the sake of simplicity, we omit the notations of $(v,w), \Pd, \pihat, \pj{j}$ from the utilities. That is, $\hat{b}^j = \thetahat^j(\pj{j}),\quad \util^{j} = \utilP^{j}(\pihat; \pj{j-1}).$
For $i\geq w$, we have: $\hat{b}^{i+1} = v = \hat{b}^{i}$. For $i < w$, we have:
\begin{align*}
    \util^{i+1} &= \Prob{p_{i+1} \leq \hat{b}^{i+1}} \E{v-p_{i+1} \mid p_{i+1} \leq \hat{b}^{i+1}} + \Prob{p_{i+1} > \hat{b}^{i+1}} \E{\util^{i+2}\mid p_{i+1} > \hat{b}^{i+1}}\\
    &\geq \Prob{p_{i+1} \leq \hat{b}^{i+1}} \E{u^{i+2} \mid p_{i+1} \leq \hat{b}^{i+1}} + \Prob{p_{i+1} > \hat{b}^{i+1}} \E{\util^{i+2}\mid p_{i+1} > \hat{b}^{i+1}}\\
    &= \E{u^{i+2}},
\end{align*}
where the probabilities and the expectations here are over $p_{i+1} \sim \Pd_{i+1}|\p_{1:i}$, and the inequality is due to the fact that
\[
p_{i+1}\leq \hat{b}^{i+1} = v - u^{i+2} \Leftrightarrow v-p_{i+1}\geq u^{i+2}.
\]
\end{proof}

Note that if the buyer buys at step $i$ then the buyer utility reflects that purchase. The partial utility from step $i+1$ onwards includes only a possible purchase after step $i$, and therefore does not include the purchase at step $i$, in case it happens. This intuitively explains why the partial utility is non-increasing.

From the last theorem, we conclude that the thresholds of the best-response buying strategy are monotonically non-decreasing in time steps.
\begin{corollary}\label{lem:buyer-monotone-in-time}
Let $\Pd$ be an optimal selling strategy, and $(v,w)$ be a buyer type. Then, for any step $i \in [\wmax]$ and realized prices $p_1, \ldots, p_i$, it holds that
\[
\Eu{p_{i+1} \sim \Pd_{i+1}|\p_{1:i}}{\thetahat^{i+1}(\p_{1:i}, p_{i+1})} \geq \thetahat^i(\p_{1:i})
\]
\end{corollary}

\subsubsection{Computing an optimal mixed selling strategy} \label{app:mixed-planning} 

In this section, we present an algorithm to find an optimal mixed selling strategy using prices from a given set $P$.

\paragraph{Overview of the algorithm.}
Our algorithm defines the optimal mixed strategy in a backward manner, starting step $\wmax$ going back to step $1$.
Given a price set $P$, for each step $i$, it computes for every possible price $p_i\in P$ and for any buyer vector $\bb^i \in \V^{\wmax}$, where $\bb^i$ will be define later, the optimal revenue $r^i(p_i, \bb^i)$ from step $i$ until the end, assuming the buyer types that reach step $i$ are according to the vector $\bb^i$, 
and requiring that the price at step $i$ is $p_i$. In addition, for each buyer of type $(v,w)$, it computes the buyer's utility $u^i_{v,w}(p_i, \bb^i)$ from step $i$ onwards given the selling strategy that we fixed from step $i$.

A \emph{buyer vector} $\bb = (b_1, \cdots, b_{\wmax})$ at step $i$ represents the set of buyer types $\set{(v,w) : v \leq b_w, w \geq i}$. 
Note that there are at most $|\V|^{\wmax}$ such vectors compared to $2^{|\V|\wmax}$ subsets of types. Due to the buyer's monotoncity, presented in \cref{crl:buyer-monotonicity}, it indeed suffices to encode the set of buyer types reaching a particular step with a vector of $\wmax$ values.

Given that the algorithm computed the results for step $i+1$, it computes the results for step $i$ as follows. For every price $p_i$,  buyer vector $\bb^i$, and  buyer vector $\bb^{i+1}$, we run a linear program that computes the probabilities on the prices at step $i+1$ which maximizes the expected revenue from step $i+1$ onward given that the price at step $i$ is $p_i$. We require two additional properties. The first is that buyers with patience $w>i$ and value at most $b^{i+1}_w$ will not buy at step $i$. The second is that for any buyer with patience $w \geq i$ and value $v\in(b^{i+1}_w,b^i_w]$ will buy at step $i$. The two conditions guarantee that if the buyer types according to $\bb^i$ reach step $i$ and observe price $p_i$, then the buyer types that continue to step $i+1$ are according to $\bb^{i+1}$. Once this is done, for each $p_i$ and $\bb^i$ we compute the maximum over expected revenue over $\bb^{i+1}$.

\paragraph{Outline of the algorithm}

In \cref{algo:mixed}, we compute the seller's revenues $r^i$ and the buyer's utilities $u^i_{v,w}$. The algorithm uses two subroutines \cref{algo:solvelp,algo:reconstruct}, which we will detail below, both of which has access to the variables of \cref{algo:mixed}, specifically, variables which are functions receiving particular price $p_i$ and buyer vector $\bb^i$.

First, we initialize all revenues $r^i$ and utilities $u^i_{v,w}$ to be zero. Next, we go over steps $i \in [\wmax]$ backwards, and compute $r^i(p_i, \bb^i)$ for each possible price $p_i \in P$ and buyer vector $\bb^i \in \V^{\wmax}$.

We iterate over each possible buyer vector $\bbf^{i+1} \in \V^{\wmax}$, satisfying $b_w^{i+1} \leq b^i_w$ for any $w \in [\wmax]$, and call Algorithm \ref{algo:solvelp} to solve a linear programming for finding the best distribution $\alphabf^{i+1}$ over the prices at step $i+1$. The objective function of this LP is the expectation over $\alphabf^{i+1}$ of the revenues from step $i+1$ onwards. \cref{cond:solvelp-1,cond:solvelp-2} ensure that $\alphabf^{i+1}$ is indeed a distribution. Due to the monotonicity of the best-response buying strategy, \cref{cond:solvelp-3} ensures that each buyer with patience $w > i$ and value at most $b^{i+1}_w$ would not prefer to buy at step $i$. \cref{cond:solvelp-4} ensures that each buyer with patience $w > i$ and value greater than $b^{i+1}_w$ would prefer to buy at step $i$. In fact, \cref{cond:solvelp-3,cond:solvelp-4} force us to consider only price distributions $\alphabf^{i+1}$ that guarantee that if the buyer vector at step $i$ is $\bb^i$ and the price is $p_i$, then the buyer vector at step $i+1$ is $\bb^{i+1}$. In Lines
(\ref{line:set-bbhat}) and (\ref{line:set-alphahat}), we choose the pair $(\alphahat^{i+1}, \bbhat^{i+1})$ that maximizes the excepted revenue from step $i+1$ onwards. Note that there is always such feasible pair.

Then, in Line (\ref{line:set-rev}), we set $r^i(p_i,\bb^i)$ to be the sum of three terms: (1) the excepted revenue from buyers whose patience is $i$ and can buy at step $i$; (2) the excepted revenue from buyers whose patience is greater than $i$ and prefer (and are able) to buy at step $i$ according to buyer vector $\bb^i$ and $\bbhat^{i+1}$; (3) the optimal future excepted revenue, $\hat{s}^{i+1}$.

In Line (\ref{line:update-utilities}), we compute $u^i_{v,w}(p_i,\bb^i)$ for each buyer type $(v,w)$. If he prefers to buy at step $i$, it is simply set to $v - p_i$. If he prefers to buy after step $i$, it is set to the expectation over $\alphahat^{i+1}$ of his utility from step $i+1$ onwards. 

Finally, in Line (\ref{line:reconstruct}), we call \cref{algo:reconstruct} to reconstruct the optimal mixed selling strategy $\Pdst$. The reconstruction works as follows. First, we find the buyer vector $\bbhat^1$ at step $1$ according to the distribution $\D$. We set the price $\hat{p_1}$ at step $1$ to be the price maximizing the revenue $r^1(p_1, \bb^1)$. Next, we iterate over any pure pricing $\p$ starting at price $\hat{p}_1$
and calculate its probability under $\Pdst$. To do so, we maintain three variables: $q$ is the accumulated probability, $\bb$ is the current buyer vector (initialized to $\bbhat^1$) and $\alphabf$ is the current price distribution. We go over steps $i$ from $1$ to $\wmax-1$, and use the optimal price distribution $\alphahat^{i+1}(p_i, \bb)$ and its corresponding buyer vector $\bbhat^{i+1}(p_i, \bb)$ to update the variables properly.

\begin{algorithm}
\caption{Optimal \textbf{mixed} selling strategy}\label{algo:mixed}
\textbf{Input:} Distribution $\D$ supported on a finite set $S_{\D} \subseteq \V \times \W \subseteq [0,1] \times [\wmax]$, and a finite set of prices $P$.
\vspace{2mm}
\\
\textbf{Declare:} for each step $i$ and buyer type $(v,w)$, let partial revenue and utility $r^i, u^i_{v,w}: P\times \V^{\wmax} \rightarrow [0,1]$, future revenues $s^{i+1}: \V^{\wmax} \rightarrow [0,1]$ and $\hat{s}^{i+1} \in [0,1]$, buyer vector $\bbhat^{i+1}: P\times \V^{\wmax} \rightarrow \V^{\wmax}$, and price distributions $\alphahat^{i+1}: P\times \V^{\wmax} \rightarrow \Delta(P)$ and $\alphabf: \V^{\wmax} \rightarrow \Delta(P)$.
\\
\textbf{Initialize:} for each step $i\in[\wmax+1]$, price $p_i \in P$ and buyer vector $\bb^i \in \V^{\wmax}$, set $r^i(p_i,\bb^i) \leftarrow 0$ and $u^i_{v,w}(p_i, \bb^i) \leftarrow 0$.
\begin{enumerate}[leftmargin=18pt,rightmargin=10pt,itemsep=1pt,topsep=3pt]
    \item For each step $i \leftarrow \wmax, \ldots, 1$, price $p_i\in P$, and buyer vector $\bb^i \in \V^{\wmax}$:
    \begin{enumerate}
        \vspace{1mm}
        \item[{\color{gray} \#}] {\color{gray} Compute the optimal distribution $\alphabf^{i+1}$ over prices for step $i+1$ and the optimal future revenue $s^{i+1}$ from step $i+1$ onwards for each buyer vector, pick the one that maximizes the revenue.}
        \item For each buyer vector $\bbf^{i+1} \in \V^{\wmax}$, satisfying $b_w^{i+1} \leq b^i_w$ for any $w \in [\wmax]$: \label{line:solve-lps}
        \\
         $\left(s^{i+1}(\bb^{i+1}),\; \alphabf^{i+1}(\bb^{i+1})\right) \leftarrow \text{SolveLP}(i, p_i, \bb^i, \bb^{i+1})$
        
        \item $\hat{s}^{i+1} \leftarrow \max_{\bb^{i+1}}\set{s^{i+1}(\bb^{i+1})}$ \label{line:set-shat}
        
        \item $\bbhat^{i+1}(p_i, \bb^i) \leftarrow \arg\max_{\bb^{i+1}}\set{s^{i+1}(\bb^{i+1})} $ \label{line:set-bbhat}
        
        \item $\alphahat^{i+1}(p_i, \bb^i) \leftarrow \alphabf^{i+1}(\bbhat^{i+1})$ \label{line:set-alphahat}
        \vspace{1mm}
        \item[{\color{gray} \#}] {\color{gray} Update the revenue for price $p_i$ and buyer vector $\bb^i$. The first term is the revenue from buyers with patience $i$, the second is the revenue from buyers with patience greater than $i$ that buy at step $i$, and the third term is the optimal future revenue.}
        \item $r^i(p_i,\bb^i) \leftarrow \Probb{(v,w) \sim \D}{w = i, p_i \leq v \leq b^i_w} \cdot p_i + \Probb{(v,w) \sim \D}{w \geq i+1, \hat{b}^{i+1}_w < v \leq b^i_w} \cdot p_i +  \hat{s}^{i+1}$ \label{line:set-rev}
        \vspace{1mm}
        \item[{\color{gray} \#}] {\color{gray} Update the utilities for price $p_i$ and buyer vector $\bb^i$.}
        \item For each buyer type $(v,w) \in S_{\D}$ such that $w\geq i$:
        \label{line:update-utilities}
        \\
        $u^i_{v,w}(p_i, \bb^i) \leftarrow \max \left\{ v-p_i, \sum_{p_{i+1}\in P} \hat{\alpha}^{i+1}_{p_{i+1}} \cdot u^{i+1}_{v, w} (p_{i+1}, \bbhat^{i+1}) \right\}$
    \end{enumerate}
    \vspace{2mm}
    \item $\Pdst \leftarrow \text{Reconstruct}(\;)$ \label{line:reconstruct}
\end{enumerate}
\textbf{Output:} Optimal mixed selling strategy, $\Pdst$.
\end{algorithm}

\begin{algorithm}
\caption{SolveLP $(i, p_i, \bb^i, \bb^{i+1})$}\label{algo:solvelp}
\textbf{Input:} step $i$, price $p_i$, buyer vector $\bb^i$, buyer vector $\bb^{i+1}$.
\begin{align}
  s^{\star}\leftarrow \max_{\alphabf \in \R^{|P|}} \quad
  &
  \sum_{p_{i+1}\in P} \alpha_{p_{i+1}} r_{i+1}(p_{i+1}, \bb^{i+1}) \nonumber
  \\
  \text{s.t.} \quad
  &
  \alpha_{p_{i+1}} \geq 0 \quad \forall p_{i+1} \in P \label{cond:solvelp-1}
  \\
  &
  \sum_{p_{i+1}\in P} \alpha_{p_{i+1}} = 1 \label{cond:solvelp-2} 
  \\
  &
  \sum_{p_{i+1}\in P} \alpha_{p_{i+1}} u^{i+1}_{v, w} (p_{i+1}, \bb^{i+1}) \geq v - p_i \quad \text{for any } w > i \text{ and } v = b^{i+1}_w \label{cond:solvelp-3}
  \\
  &
  \sum_{p_{i+1}\in P} \alpha_{p_{i+1}} u^{i+1}_{v, w} (p_{i+1}, \bb^{i+1}) \leq v - p_i \quad \text{for any } w > i \text{ and } v =\min\set{x>b^{i+1}_w : (x,w) \in S_{\D}} \label{cond:solvelp-4}
  &&
 \end{align}
\textbf{Output:} the pair $(s^{\star}, \alphabf^{\star})$, where $\alphabf^{\star}$ is the solution for the above linear programming. If there is no solution, set $s^{\star} \leftarrow 0$ and $\alphabf^{\star} \leftarrow \alphabf$, where $\alphabf$ is some arbitrary distribution.

\end{algorithm}

\begin{algorithm}
\caption{Reconstruct the optimal selling strategy}\label{algo:reconstruct}
\textbf{Declare:} let mixed selling strategy $\Pd \in \Delta(P^{\wmax})$, probability $q \in [0,1]$, buyer vector $\bb \in \V^{\wmax}$, and price distribution $\alphabf \in \Delta(P)$.
\\
\textbf{Initialize:} for each pure pricing $\p \in P^{\wmax}$,
set $\Pd[\p] \leftarrow 0$.
\\
\begin{enumerate}[leftmargin=18pt,rightmargin=10pt,itemsep=1pt,topsep=3pt]
    \item Set the buyer vector at step $1$ to be $\bbhat^1$ such that $\hat{b}^1_w = \max\set{v : (v,w)\in S_{\D}}$ for any $w$. 
    \label{rline:set-bb-one}
    
    \item $\hat{p}_1 = \argmax_{p_1} r^1(p_1, \bbhat^1)$. 
    \label{rline:set-p-one}
    
    \item For each pure pricing $\p \in \set{\hat{p_1}}\times P^{\wmax-1}$: 
    \label{rline:compute-probs}
    
    \begin{enumerate}
        \item $q \leftarrow 1$
        \item $\bb \leftarrow \bbhat^1$
        \item For each step $i \leftarrow 1, \ldots, \wmax-1$:
        \begin{enumerate}
            \item $\alphabf \leftarrow \alphahat^{i+1}(p_i, \bb)$
            \item $q \leftarrow q \cdot \alpha_{p_{i+1}}$
            \item $\bb \leftarrow \bbhat^{i+1}(p_i, \bb)$
        \end{enumerate}
        \item $\Pd[\p] \leftarrow q$
    \end{enumerate}
\end{enumerate}
\textbf{Output:} reconstructed mixed selling strategy $\Pd$.
\end{algorithm}

\thmmixedplanning* 

To prove the correctness of \cref{algo:mixed} we first show that  the total revenue of the resulting mixed strategy $\Pdst$ is $r^1(\hat{p}_1, \bbhat^1)$, computed by the algorithm. To that end, we prove by backward induction that for each step $i$, if the price is $p_i$ and the buyer vector is $\bb^i$, then the following holds: (1) the revenue from step $i$ onwards is $r^i(p_i, \bb^i)$, (2) the utility from step $i$ onwards for a buyer type $(v,w)$ is $u^i_{v,w}(p_i, \bb^i)$, (3) the buyer vector at step $i+1$ is $\bbhat^{i+1}(p_i, \bb^i)$. To complete the the correctness, we further show that total revenue of any other mixed selling strategy does not exceed that of $\Pdst$. We prove by backward induction that for any mixed selling strategy $\Ptilde$ with buyer vectors $\bbtilde^i(\pj{i-1})$, it holds for any history $\p_{1:i}=(p_1, \ldots, p_i)$ that $r^i(p_i, \bbtilde^i) \geq \rtilde^i(\pj{i}, \bbtilde^i)$, where $\bbtilde^i = \bbtilde^i(\pj{i-1})$, and $\rtilde^i(\pj{i}, \bb^i)$ is the revenue of $\Ptilde$ from step $i$ onwards, given the history $\pj{i}$ and the buyer vector $\bb^i$ at step $i$.

\begin{proof}[of \cref{thm:mixed-planning}]
We prove the correctness and running time of \cref{algo:mixed}.
\paragraph{Correctness.}

Assume the seller plays according to the algorithm's output selling strategy $\Pdst$, and the buyer plays according to his best-response threshold buying strategy, presented in $\cref{thm:threshold-equiv}$.

First, we show that the revenue of $\Pdst$ is $r^1(\hat{p}_1, \bbhat^1)$, computed by the algorithm. We prove by backward induction that for each step $i$, if the price is $p_i$ and the buyer vector is $\bb^i$, then the following holds: (1) the revenue from step $i$ onwards is $r^i(p_i, \bb^i)$, (2) the utility from step $i$ onwards for a buyer type $(v,w)$ is $u^i_{v,w}(p_i, \bb^i)$, (3) the buyer vector at step $i+1$ is $\bbhat^{i+1}(p_i, \bb^i)$. For $i=\wmax+1$, the only possible buyer vector is $\bb^{\wmax+1} = (0, \ldots, 0)$, and indeed by initialization, we have for any price $p \in P$ that $r_{\wmax+1}[p, (0,\ldots,0)] =u^{\wmax+1}_{v,w}[p, (0,\ldots,0)]= 0$. Assume the inductive hypothesis holds for any step $j > i$, and prove for $i$. Let $p_i$ and $\bb^i$ be the price and buyer vector at step $i$, respectively. Let $\bb^{i+1}$ be the buyer vector appeared after $(p_i, \bb^i)$ according to the strategy $\Pdst$. For simplicity, we denote $\alphahat^{i+1} = \alphahat^{i+1}(p_i, \bb^i), \bbhat^{i+1} = \bbhat^{i+1}(p_i, \bb^i)$. By the induction assumption (2), the excepted utility of a buyer of type $(v,w)$ from step $i+1$ onwards is
\[
\sum_{p_{i+1}\in P} \alphahat^{i+1}_{p_{i+1}} u^{i+1}_{v, w} (p_{i+1}, \bbhat^{i+1})
\]
Hence, we will see that $\bb^{i+1} = \bbhat^{i+1}$. Indeed, since according to the algorithm, the pair $(\alphahat^{i+1}, \bbhat^{i+1})$ is a feasible solution to the linear program (\cref{algo:solvelp}), so it holds \crefrange{cond:solvelp-1}{cond:solvelp-4}. In particular, by \cref{cond:solvelp-3}, a buyer with patience $w > i$ and value $v = \hat{b}^{i+1}_w$ would prefer not to buy at step $i$. Due to the monotonicity
of the buying strategy (\cref{thm:buyer-monotone-in-types}), any buyer with patience $w > i$ and value $v < \hat{b}^{i+1}_w$ would also prefer not to buy at step $i$. By \cref{cond:solvelp-4}, a buyer with patience $w > i$ and value $v =\min\{x\in P: x>\hat{b}^{i+1}_w\}$ would prefer to buy at step $i$. By monotonicity, any buyer with patience $w > i$ and value $v > \hat{b}^{i+1}_w$ would also prefer to buy at step $i$. Thus we get that $\bb^{i+1} = \bbhat^{i+1}$. Therefore, the excepted utility of a buyer of type $(v,w)$ from step $i$ onwards is
\[
u^i_{v,w}(p_i, \bb^i) = \max\set{v-p_i, \sum_{p_{i+1}\in P} \alphahat^{i+1}_{p_{i+1}} u^{i+1}_{v, w} (p_{i+1}, \bbhat^{i+1})}.
\] 
By the induction assumption (3), the excepted revenue from step $i+1$ onwards is
\[
\hat{s}^{i+1}(p_i, \bb^i) = \sum_{p_{i+1}\in P} \alphahat^{i+1}_{p_{i+1}} r^{i+1} (p_{i+1}, \bbhat^{i+1}),
\]
so the revenue from step $i$ onwards is
\[
r^i(p_i,\bb^i) = \Probb{(v,w) \sim \D}{w = i, p \leq v \leq b^i_w} \cdot p_i + \Probb{(v,w) \sim \D}{w \geq i+1, \hat{b}^{i+1}_w < v \leq b^i_w} \cdot p_i +  \hat{s}^{i+1}(p_i, \bb^i),
\]
as required.

Now, let $\Ptilde$ be some mixed selling strategy. Denote by $\bbtilde^i(\pj{i-1})$ the buyer vector at step $i$ after history $\pj{i-1}=(p_1, \ldots, p_{i-1})$, relative to strategy $\Ptilde$. Also, denote by $\rtilde^i(\pj{i}, \bb^i)$ the revenue of $\Ptilde$ from step $i$ onwards after history $\pj{i}=(p_1, \ldots, p_i)$ when the buyer vector at step $i$ is $\bb^i$. We show that for any step $i \in [\wmax]$ and history $\pj{i}$, it holds that $r^i(p_i, \bbtilde^i) \geq \rtilde^i(\pj{i}, \bbtilde^i)$, where $\bbtilde^i = \bbtilde^i(\pj{i-1})$. We prove it by backward induction on step $i$. For $i=\wmax+1$, by initialization, $r^{\wmax+1}(p_{\wmax+1}, (0, \ldots, 0)) = 0 = \rtilde^{\wmax+1}(\p, (0, \ldots, 0)))$ for any $\p \in P^{\wmax+1}$. Assume the inductive hypothesis holds for any step $j > i$, and prove for $i$. Let $\pj{i}=(p_1, \ldots, p_i)$ be a history until step $i$. Fix buyer vectors $\bbtilde^i = \bbtilde^i(\pj{i-1})$ and $\bbtilde^{i+1} = \bbtilde^i(\pj{i})$. The algorithm iterates over any buyer vector $\bb^{i+1}$ and picks the buyer vector $\bbhat^{i+1}$ that maximizes the future revenue from step $i+1$ onwards, $s^{i+1}(\bb^{i+1}) = \max_{\alphabf} \Eb{p_{i+1} \sim \alphabf}{r^{i+1}(p_{i+1}, \bb^{i+1})}$, where $\alphabf$ is subject to \crefrange{cond:solvelp-1}{cond:solvelp-4}. Let $\alphatilde^{i+1}$ be the mixture over prices at step $i+1$ according to $\Ptilde$ after history $\pj{i}$. The distribution $\alphatilde^{i+1}$ holds \cref{cond:solvelp-1,cond:solvelp-2}. In addition, since the buyer plays a best-response strategy, $\alphatilde^{i+1}$ and $\bbtilde^{i+1}$ hold \cref{cond:solvelp-3,cond:solvelp-4}. Hence,
\[
s^{i+1}(\bbhat^{i+1}) = \max_{\alphabf,\bb^{i+1}} \Eu{p_{i+1} \sim \alphabf}{r^{i+1}(p_{i+1}, \bb^{i+1})} \geq \Eu{p_{i+1} \sim \alphatilde}{r^{i+1}(p_{i+1}, \bbtilde^{i+1})} \geq \Eu{p_{i+1} \sim \alphatilde}{\rtilde^{i+1}(\pj{i+1}, \bbtilde^{i+1})},
\]
where the last inequality follows from the induction assumption. So we get that the future revenue from step $i+1$ onwards computed by the algorithm is greater than or equal to that under selling strategy $\Ptilde$. Finally, the revenue of both strategies obtained from buyers that buy at step $i$ is the same, since it is computed relative to the same price $p_i$ and buyer vector $\bbtilde^i$, so we get that $r^i(p_i, \bbtilde^i) \geq \rtilde^i(\pj{i}, \bbtilde^i)$, as required.

\paragraph{Running time.}
The initialization part takes $\O(|\V|^{\wmax} |P| \wmax)$. The SolveLP function (\cref{algo:solvelp}) involves solving linear programming model with $|P|$ variables and thus its running time is polynomial in $|P|$ (using the Ellipsoid method, for instance). So each iteration in (\ref{line:solve-lps}) takes $\poly(|P|)$, and in total (\ref{line:solve-lps}) takes $\O(\poly(|P|) |\V|^{\wmax})$. Lines (\ref{line:set-shat}) to (\ref{line:set-alphahat}) take $\O(|\V|^{\wmax})$. The revenue update in Line (\ref{line:set-rev}) takes $\O(1)$. The update of utilities in Line (\ref{line:update-utilities}) takes $\O(|\V||P| \wmax)$. So each outer iteration takes $\O(\poly(|P|)|V|^{\wmax}\wmax)$, and the whole loop takes $\O(\poly(|P|)|V|^{2\wmax}\wmax^2)$. 

For the reconstruction part which appears in \cref{algo:reconstruct}, Line (\ref{rline:set-bb-one}) takes $\O(|V|\wmax)$, Line (\ref{rline:set-p-one}) takes $\O(|P|)$, and the loop in Line (\ref{rline:compute-probs}) takes $\O(|P|^{\wmax}\wmax)$.
Hence, the reconstruction takes $\O(|V|\wmax + |P|^{\wmax}\wmax)$.

In total, \cref{algo:mixed} runs in $\O(\poly(|P|^{\wmax})|V|^{2\wmax}\wmax^2)$.
\end{proof}

\subsection*{Handling continuous values}

We would also like for mixed strategies to address the issue of continuous support. Assume the set of buyer's values $\V$ is $[0,1]$. We take again the possible prices to be the discretized set $P_\eps=\{0,\eps,2\eps, \ldots, 1\}$, and guarantee an $\eps\wmax$-optimal selling strategy for $\eps>0$.

The following lemma shows that the discretization gives a good approximation even in the mixed case.

\begin{lemma}\label{lem:eps-w-optimal-mixed}
Assume the buyer's value set is $[0,1]$.
Let $\hat{\Pd}$ be an optimal mixed selling strategy with respect to the mixed strategies over a discretization $P_\eps = \{0,\eps,2\eps, \ldots, 1\}$ of $[0,1]$, for any $\eps > 0$. Then, $\hat{\Pd}$ is an $\eps\wmax$-optimal with respect to the optimal mixed strategies over $[0,1]$.
\end{lemma}

In the proof of \cref{lem:eps-w-optimal-mixed}, we modify an optimal mixed selling strategy ensuring that it offers prices only from $P_\eps$. We first decrease each price $p_i$, on which the strategy mixes at step $i$, by $(\wmax-i)\eps$, and then decrease the result down to the nearest price $i\eps\in P_\eps$. This modification guarantees that at any step $i$, the buyer would certainly get the maximum discount he may receive in the future, and even more. From this it follows that if the buyer preferred to buy at step $i$ before the change, he still prefers to buy at step $i$ or before. Moreover, under this change, the buyer’s utility may increase by at most $\eps\wmax$. Hence we obtain that the revenue of the resulting mixed strategy is lower than that of the original by at most $\eps\wmax$. The main issue that our mapping takes care is that a buyer that bought at step $i$ will not prefer to buy at a later step $j>i$. This is important since otherwise the revenue might decrease and we cannot bound this potential decrease.

\begin{proof}[of \cref{lem:eps-w-optimal-mixed}]
By \cref{lem:mixed-over-value-range}, there exists an optimal mixed selling strategy $\Pdst$ using prices only from $[0,1]$. We modify  $\Pdst$ to ensure that it offers prices only from $P_\eps$. Define the following transformation $f: [0,1]^{\wmax} \rightarrow [0,1]^{\wmax}$ over the set of pure strategies:
\[
[f(\p)]_i = \eps \cdot \left\lfloor\frac{p_i - (\wmax-i)\eps}{\eps}\right\rfloor
\]
Denote by $\Pd'$ the mixed strategy obtained from invoking $f$ on $\p \sim \Pdst$, i.e., for any $\z=f(\p)$ let $Z=\{\p:f(\p)=\z\}$ and we set $\Pd'(\z)=\Pdst(Z)$. To show that $r(\Pdst) - r(\Pd') \leq \eps\wmax$, by \cref{lem:sufficient-conds-for-greater-rev}, it suffices to prove that for any buyer type $(v,w)$ after any history $\p_{1:i}$ of length $i\in [\wmax]$, both of the following conditions are met: (1) if the buyer preferred to buy at step $i$, he still buys at this step; (2) his utility from step $i$ onwards may increase by at most $\eps(\wmax-i+1)$. We prove it by backward induction on the history length $i$.
For $i=\wmax$, since this is the final step, the buyer’s decision is based only on his value. According to the transformation $f$, the realized price $p_{\wmax}$ may only decrease and by at most $\eps$. Hence, if the buyer could buy at step $\wmax$, he can still buy at this step, and also his utility at step $\wmax$ may increase by at most $\eps = \eps(\wmax-\wmax+1)$. Assume the inductive hypothesis holds for any $j > i$ and prove for step $i$. Let $(v,w)$ be a buyer type with patience $w \geq i$. Note that the transformation $f$ insures that the realized price $p_i$ at 
step $i$ decreases by at least $\eps(\wmax-i)$ and by at most $\eps(\wmax-i+1)$, i.e., $\eps(\wmax-i) \leq p_i - p'_i \leq \eps(\wmax-i+i)$. Hence, the buyer's utility at step $i$ for making an immediate purchase would increase  by at least $\eps(\wmax-i)$ and by at most $\eps(\wmax-i+1)$. By induction assumption, the utility from step $i+1$ onwards may increase by at at most $\eps(\wmax-(i+1)+1) = \eps(\wmax-i)$. Thus, if the buyer preferred to buy at step $i$ before the change, he would still buy at this step as well. Moreover, we get that  his utility from step $i$ onwards may increase by at most $\eps(\wmax-i+1)$, as required. This completes the proof of the inductive hypothesis.

Therefore, by \cref{lem:sufficient-conds-for-greater-rev}, the excepted loss of this modification is at most $\eps\wmax$, that is, $\rev(\Pdst) - \rev(\Pd') \leq \eps\wmax$. Recall that $\Phat$ is the optimal mixed strategy with respect to mixed strategies over $P_\eps$. Clearly, $\rev(\Pdst)\geq\rev(\Phat)$ and $\rev(\Phat)\geq\rev(\Pd')$. Therefore,
\[
\rev(\Pdst) - \rev(\Phat) \leq \rev(\Pdst) - \rev(\Pd') \leq \eps\wmax.
\]

\end{proof}

Combining \cref{lem:eps-w-optimal-mixed,thm:mixed-planning}, we conclude,

\begin{lemma}\label{thm:eps-w-optimal-mixed-planning}
There exists an algorithm that for any $\eps>0$ and a distribution $\D$ over $[0,1] \times [\wmax]$, returns an $\eps\wmax$-optimal mixed strategy, and runs in time $\O(\poly(1/\eps^{\wmax})\wmax^2)$.
\end{lemma}

\begin{proof}
Let $\eps > 0$. We discretize the buyer's value set $\V$ to obtain a finite set of values $\V_\eps$ with $\O(1/\eps)$ values and run \cref{algo:mixed} with the price set $P = \V_\eps$, so the theorem follows from \cref{thm:mixed-planning,lem:eps-w-optimal-mixed}.
\end{proof}

\section{Proofs for Section \ref{sec:learning}}\label{app:sec_learning}

\begin{subsection}{Proofs of \cref{subsec:background-leaning}}\label{app:background-learning}
The following two claims follow from the definition of fat-shattering dimension. We make use of it in this Section.
\begin{lemma}\label{lem:r-non-neg}
Let $\F$ be a  class of non-negative real-valued functions from input space $\Zc$ and $\gamma > 0$. If $\cb \in \R^m$ witness that $S = \{z_1, \ldots, z_m\} \subseteq \Zc$ is $\gamma$-shattered by $\F$, then it must hold that $\cb$ is also non-negative.
\end{lemma}

\begin{proof}
If there is $j \in [m]$ with $c_j < 0$, then take $\sigmab \in \{-1, 1\}^m$ such that $\sigma_j = -1$, so we have for all $f \in \F$:
\[
f(z_j) \geq 0 > c_j \geq c_j - \gamma \Rightarrow f(z_j) > c_j - \gamma
\]
which contradicts that $\cb$ witness the $\gamma$-shattering.
\end{proof}

\begin{lemma} \label{lem:zero-fat}
Let $\F$ be a  class of functions $f: \Zc \rightarrow [a,b]$ and $\gamma > 0$. If $\gamma > \frac{b-a}{2}$, then $\fatgamma{\F} = 0$.
\end{lemma}

\begin{proof}
Let $S = \left\{z_1\right\} \subseteq \Zc$ and $\gamma > \frac{b-a}{2}$. Assume by contradiction that $S$ is $\gamma$-shattered by $\F$ with some witness $c_1 \in \R$. Hence, there are two functions $f_+, f_- \in \F$ such that
\[
f_+(z_1) \geq c_1 + \gamma, \quad f_-(z_1) \leq c_1 - \gamma
\]
Since $f_+(z_1), f_-(z_1) \in [a,b]$ and $\gamma > \frac{b-a}{2}$, we have:
\begin{align*}
    c_1 &\geq f_-(z_1) + \gamma > a + \frac{b-a}{2} = \frac{a+b}{2}\\
    c_1 &\leq f_+(z_1) - \gamma < b - \frac{b-a}{2} = \frac{a+b}{2}
\end{align*}
so we get a contradiction.
\end{proof}

The following is a well known uniform convergence theorem for classes with finite fat-shattering for all $\gamma>0$.

\uctheorem*

\begin{proof}
The proof follows from some classic results and nicely summarized in \citet[Section 8]{vershynin2018high}.
First,
\[
\sup_{f \in \F} \abs{\frac{1}{m} \sum_{i=1}^m f(z_i) - \Eu{z \sim \D}{f(z)}} \lesssim 
\cR_{m}(\F) + \sqrt{\frac{\log \frac{1}{\delta}}{m}},
\]
where $\cR_{m}(\F)$ is the Rademacher Complexity of the class $\F$. See, e,g, \citet[Section 3]{mohri2018foundations}.

By chaining and Dudley’s Entropy Integral \citep{dudley1967sizes}, we have
\[
\cR_{m}(\F) 
\lesssim
\int_{0}^{\infty} \sqrt{\log\mathcal{N}(\gamma,\F,\mathcal{L}_2(\D_m))} \,d\gamma,
\]
where $\mathcal{N}(\gamma,\F,\mathcal{L}_2(\D_m))$ is the covering number of $\F$ at scale $\gamma$ with respect to $\mathcal{L}_2$ metric.
 \citet{mendelson2003entropy} proved that $\mathcal{N}(\cdot)\lesssim \paren{\frac{1}{\gamma}}^{\fat_{\gamma}(\F)}$, this will incur a superfluous $\log$ factor. \citet{rudelson2006combinatorics} proved
\[
\int_{0}^{\infty} \sqrt{\log\mathcal{N}(\gamma,\F,\mathcal{L}_2(\D_m))}
\lesssim
\int_{0}^{\infty} \sqrt{\fatgamma{\F}} \,d\gamma,
\]
and this ends up the proof.

\citet{block2021majorizing} proved a similar result for a sequential version of Rademacher Complexity. 
\end{proof}

\end{subsection}

\subsection{Proofs for \cref{subsec:sample-complexity-pure}}\label{app:sample-complexity-pure}

\subsection*{Sample complexity: upper bounds}

Denote the class of pure selling strategies with non-increasing prices by
\[
\npureset = \set{\p=(p_1,\ldots,p_{\wmax}) \in [0,1]^{\wmax}:p_1 \geq \cdots \geq p_{\wmax}}.
\]
Recall the class of revenue functions associated to $\npureset$,
\[
\Rpn = \set{z \mapsto r(\p,z):\; \p \in \npureset},
\]
where for non-increasing pricing $\p$ we have $r(\p, (v,w)) = \ind{v \geq p_w} \cdot p_w$.

\lemfatgain*

We prove a more refined bound.
\begin{lemma}
For every $\gamma>0$,
\[
\fatgamma{\Rpn} \leq
\begin{cases}
    2\wmax, & \gamma \in (0, \frac{1}{4}) \\
    \wmax, & \gamma \in [\frac{1}{4}, \frac{1}{2}] \\
    0, & \gamma > \frac{1}{2}.
\end{cases}
\]
For every $0 < \alpha < \frac{1}{\wmax-1}$ and $0 < \gamma < \frac{1}{2} - \frac{\wmax-1}{2} \alpha$,
\[
\fatgamma{\Rpn} \geq \wmax,
\]
we conclude that
the pseudo-dimension of the function class $\Rpn$ is $\Theta(\wmax)$. 
\end{lemma}

We prove the following claims.

\begin{lemma}\label{lemma:gain-fat-lower-bound}
For every $0 < \alpha < \frac{1}{{\wmax}-1}$ and $0 < \gamma < \frac{1}{2} - \frac{{\wmax}-1}{2} \alpha$,
\[
\fatgamma{\Rpn} \geq {\wmax}.
\]
\end{lemma}

\begin{lemma}\label{lemma:gain-fat-upper-bound}
For every $\gamma > 0$,
\[
\fatgamma{\Rpn} \leq
\begin{cases}
    2{\wmax}, & \gamma \in (0, \frac{1}{4}) \\
    {\wmax}, & \gamma \in [\frac{1}{4}, \frac{1}{2}] \\
    0, & \gamma > \frac{1}{2}.
\end{cases}
\]
\end{lemma}

In order to prove \cref{lemma:gain-fat-upper-bound} we use the following.
Define the projection of $\Rpn$ on patience $w$ by
\[
\Rc_w = \{ v \mapsto f(v,w) \mid f \in \Rpn \}.
\]
\begin{lemma} \label{lem:fatd}
For every $\gamma > 0$ and $w \in [{\wmax}]$,
\[
\fatgamma{\Rc_w} = 
\begin{cases}
    2, & \gamma \in (0, \frac{1}{4}) \\
    1, & \gamma \in [\frac{1}{4}, \frac{1}{2}] \\
    0, & \gamma > \frac{1}{2}.
\end{cases}
\]
\end{lemma}

Combining \cref{lemma:gain-fat-lower-bound} and \cref{lemma:gain-fat-upper-bound}, we can conclude \cref{lem:fat-gain}.

\begin{proof}[of \cref{lemma:gain-fat-lower-bound}]
Let $0 < \alpha < \frac{1}{{\wmax}-1}$ and $0 < \gamma < \frac{1}{2} - \frac{{\wmax}-1}{2} \alpha$. Take the set:
\[
S = \set{(v_i, w_i):i \in [{\wmax}]} = \set{(2 \gamma + ({\wmax}-i) \alpha, i):i \in [{\wmax}]} \subseteq [0,1] \times [{\wmax}]
\]
and the witness $\cb \in \R^{\wmax}$ defined by $c_i = \gamma + ({\wmax}-i) \alpha$ for every $i \in [{\wmax}]$. Let $\sigmab \in \{-1, +1\}^{\wmax}$. Let $j$ be the maximal index such that $\sigma_i = -1$ for all $i \leq j$ or 0 if none; and let $k$ be the minimal index such that $\sigma_i = -1$ for all $i \geq k$ or ${\wmax}+1$ if none. In other words, $j$ is the ending index of the maximal prefix of $-1$'s in $\sigmab$, and $k$ is the starting index of the maximal suffix of $-1$'s in $\sigmab$. Note that the only case where $j > k$ is when $\sigmab = \{-1\}^{\wmax}$; in this case we have $j = {\wmax} > 1 = k$. If $j < k$, for every $i \in (j, k)$, denote by $\ell_i$ the maximal index $\leq i$ such that $\sigma_{\ell_i} = +1$. Define pricing $\p^{\sigmab} \in [0,1]^{\wmax}$ as follows:
\[
p^{\sigmab}_i = 
\begin{cases}
1, & 1 \leq i \leq j\\
c_{\ell_i} + \gamma, & j < i < k\\ 
0, & 1 < k \leq i \leq {\wmax}. \\
\end{cases}
\]
Since $1 > c_1 + \gamma > \ldots > c_{\wmax} +\gamma> 0$, $\p^{\sigmab}$ is monotonically non-increasing pricing, and thus $\p^{\sigmab} \in \npureset$. Let $i \in [{\wmax}]$. Consider the following cases:
\bi

\item $1 \leq i \leq j$: we have $p^{\sigmab}_i = 1$ and $\sigma_i = -1$. Hence,
\[
\ind{v_i \geq p^{\sigmab}_i} \cdot p^{\sigmab}_i = \ind{2 \gamma + ({\wmax}-i) \alpha \geq 1} \cdot 1 = 0 \leq ({\wmax}-i)\alpha = \gamma + ({\wmax}-i)\alpha - \gamma = c_i - \gamma.
\]

\item $1 < k \leq i \leq \wmax$: we have $p^{\sigmab}_i = 0$ and $\sigma_i = -1$. Hence,
\[
\ind{v_i \geq p^{\sigmab}_i} \cdot p^{\sigmab}_i = \ind{2 \gamma + ({\wmax}-i) \alpha \geq 0} \cdot 1 = 1 \geq 2\gamma + ({\wmax}-i)\alpha = \gamma + ({\wmax}-i)\alpha + \gamma = c_i + \gamma.
\]

\item $j < i < k$: we have $p^{\sigmab}_i = c_{\ell_i} + \gamma$. If $\sigma_i = +1$, it hold that $\ell_i = i$ and thus
\begin{align*}
    \ind{v_i \geq p^{\sigmab}_i} \cdot p^{\sigmab}_i &= \ind{2 \gamma + ({\wmax}-i) \alpha \geq c_i + \gamma} \cdot (c_i + \gamma)
    = \ind{2 \gamma + ({\wmax}-i) \alpha \geq 2 \gamma + ({\wmax}-i) \alpha} \cdot (c_i + \gamma)\\
    &= c_i + \gamma \geq c_i + \gamma.
\end{align*}
However, if $\sigma_i = -1$, it holds that $\ell_i < i$ and thus
\begin{align*}
    \ind{v_i \geq p^{\sigmab}_i} \cdot p^{\sigmab}_i &= \ind{2 \gamma + ({\wmax}-i) \alpha \geq c_{\ell_i} + \gamma} \cdot (c_{\ell_i} + \gamma)
    = \ind{2 \gamma + ({\wmax}-i) \alpha \geq 2 \gamma + ({\wmax}-\ell_i) \alpha} \cdot (c_{\ell_i} + \gamma) \\
    &= 0 \leq ({\wmax}-i)\alpha = \gamma + ({\wmax}-i)\alpha - \gamma = c_i - \gamma.
\end{align*}
\ei
\end{proof}

\begin{proof}[of \cref{lemma:gain-fat-upper-bound}]
Since the image of each function in $\Rpn$ is $[0,1]$, by \cref{lem:zero-fat} we get that for every $\gamma > \frac{1}{2}$, it holds that $\fatgamma{\Rpn} = 0$.

Now, let $\gamma \leq \frac{1}{2}$. We show that $\fatgamma{\Rpn} \leq k{\wmax}$, where $k=2$ if $\gamma\in(0,1/4)$, or $k=1$ if $\gamma \in [1/4,1/2]$. To do so, we use the pigeonhole principle to argue that if there exists a set of size $k\wmax+1$ shattered by $\Rpn$, we can shatter a set of size $k+1$ by some projected class $\Rc_w$.

Denote $m = k{\wmax}+1$.  Assume by contradiction that there exists a set $S = \{(v_1, w_1), \ldots, (v_m, w_m) \} \subseteq [0,1] \times [{\wmax}]$, which is $\gamma$-shattered by $\Rpn$ with some witness $\cb \in \R^m$.
For each $w \in [{\wmax}]$, denote:
\[
I_w = \set{i \in [m]:w_i = w}, \quad S_w = \set{v_i: i \in I_w}, \quad m_w = \abs{S_w}.
\]
By the pigeonhole principle, there exists a patience $w$ such that $m_w \geq k+1$. To see that $(c_i)_{i \in I_w}$ witness the $\gamma$-shattering of $S_w$ by $\Rc_w$, let ${\sigmab}^w \in \{-1,1\}^{m_w}$. Define the position of $i \in I_w$ in $I_w$ by $a(i) = \abs{\set{j \in I_w:j \leq i}}$. Then, define ${\sigmab} \in \{-1,1\}^m$ by $\sigma_i = \sigma^w_{a(i)}$ if $i \in I_w$ and 1 otherwise. Since $S$ is $\gamma$-shattered by $\Rpn$ with witness $\cb$, there exists $f^{\sigmab} \in \Rpn$ such that
\[
\forall i \in [m]. \; 
\begin{cases}
	f^{\sigmab}(v_i, w_i) \geq c_i + \gamma, & \sigma_i = 1\\
    f^{\sigmab}(v_i, w_i) \leq c_i - \gamma, & \sigma_i = -1.
 \end{cases}
\]
Define $f_w^{\sigmab}(v) = f^{\sigmab}(v, w)$. Clearly, $f_w^{\sigmab} \in \Rc_w$. Moreover, we have:
\[
\forall i \in I_w. \; 
\begin{cases}
	f_w^{\sigmab}(v_i) = f^{\sigmab}(v_i, w) = f^{\sigmab}(v_i, w_i) \geq c_i + \gamma, & \sigma_i = 1\\
    f_w^{\sigmab}(v_i) = f^{\sigmab}(v_i, w) = f^{\sigmab}(v_i, w_i) \leq c_i - \gamma, & \sigma_i = -1,
 \end{cases}
\]
so we get that $S_w$ is $\gamma$-shattered by $\Rc_w$. By \cref{lem:fatd}, there does not exist a set of size $k+1$ shattered by, so we obtain a contradiction, as required.
\end{proof}

\begin{proof} [of \cref{lem:fatd}]
Let $w \in [\wmax]$ be some patience and $\gamma > 0$. We have the following cases.
\begin{enumerate}
    \item $\gamma > \frac{1}{2}$. \\
        Since the image of each function in $\Rc_w$ is $[0,1]$, by \cref{lem:zero-fat}, for any $\gamma > \frac{1}{2}$, it holds that $\fatgamma{\Rc_w} = 0$.
        
    \item $\gamma \in [\frac{1}{4}, \frac{1}{2}]$. \\
    \begin{enumerate}
        \item $\fatgamma{\Rc_w} \geq 1$ for $\gamma \in (0, \frac{1}{2}]$. Take the set $S = \{v_1\} = \left\{ 2\gamma \right\} \subseteq [0,1]$ and the witness $c_1 = \gamma$. Consider each $\sigma \in \{-1, 1\}$:
            \be
                \item $\sigma = 1$: Take $\p^\sigma \in \npureset$ such that $p^\sigma_w = 2\gamma$, so we have:
                \begin{align*}
                    \ind{v_1 \geq p^{\sigmab}_w} \cdot p^{\sigmab}_w &= \ind{2\gamma \geq 2\gamma} \cdot 2\gamma = 2\gamma = \gamma + \gamma = c_1 + \gamma \geq c_1 + \gamma.
                \end{align*}
        
                \item $\sigma = -1$: Take $\p^{\sigma} \in \npureset$ such that $p^\sigma_w =0$, so we have:
                    \begin{align*}
                        \ind{v_1 \geq p^{\sigmab}_w} \cdot p^{\sigmab}_w &= \ind{2\gamma \geq 0} \cdot 0 = 0 = \gamma - \gamma = c_1 - \gamma \leq c_1 - \gamma,
                    \end{align*}
            \ee
        as required.
        \item $\fatgamma{\Rc_w} \leq 1$ for $\gamma \geq \frac{1}{4}$. Assume by contradiction that there exists a set $S = \{v_1, v_2\} \subseteq [0,1]$, where $v_1 < v_2$, which is $\gamma$-shattered by $\Rc_w$ with some witness $\cb \in \R^2$. Using \cref{lem:r-non-neg}, we get that $\cb$ is non-negative.
        By the $\gamma$-shattering assumption, there exists $\p^{++}, \p^{+-}, \p^{-+} \in \npureset$ such that
        \begin{align}
            \ind{v_1 \geq p^{++}_w} \cdot p^{++}_w &\geq c_1 + \gamma \label{v1++}, \\
            \ind{v_2 \geq p^{++}_w} \cdot p^{++}_w &\geq c_2 + \gamma \label{v2++}, \\
            \ind{v_1 \geq p^{+-}_w} \cdot p^{+-}_w &\geq c_1 + \gamma \label{v1+-}, \\
            \ind{v_2 \geq p^{+-}_w} \cdot p^{+-}_w &\leq c_2 - \gamma \label{v2+-}, \\
            \ind{v_1 \geq p^{-+}_w} \cdot p^{-+}_w &\leq c_1 - \gamma \label{v1-+}, \\
            \ind{v_2 \geq p^{-+}_w} \cdot p^{-+}_w &\geq c_2 + \gamma. \nonumber
        \end{align}
        Since $c_1, c_2 \geq 0$ and $\gamma > 0$, we get from (\ref{v1++}) and (\ref{v2++}) that $\ind{v_1 \geq p^{++}_w} = \ind{v_2 \geq p^{++}_w} = 1$, so we have:
        \begin{align*}
            c_1 + \gamma &\leq p^{++}_w \leq v_1,\\
            c_2 + \gamma &\leq p^{++}_w \leq v_2.
        \end{align*}
        Hence,
        \begin{align}
            c_2 + \gamma \leq v_1 < v_2. \label{ineq-1}
        \end{align}
        Now, note that the revenue function $r$ is monotonic in the buyer's value. Thus, $r(\p^{+-}, (v_1, w)) \leq r(\p^{+-}, (v_2,w))$, and from (\ref{v1+-}) and (\ref{v2+-}) we get:
        \begin{align}
            c_1 + \gamma \leq c_2 - \gamma \Rightarrow c_2 \geq c_1 + 2\gamma. \label{ineq-2}
        \end{align}
        Finally, since the revenue function $r$ is non-negative, from (\ref{v1-+}) we get that $c_1 \geq \gamma$. Combining this with inequalities (\ref{ineq-1}) and (\ref{ineq-2}) we obtain:
        \[
        v_2 > c_2 + \gamma \geq c_1 + 2\gamma + \gamma = c_1 + 3\gamma \geq \gamma + 3\gamma = 4\gamma.
        \]
        But since $v_2 \in [0,1]$, we get that $\gamma < \frac{1}{4}$ which is a contradiction.
    \end{enumerate}
    
    \item $\gamma \in (0, \frac{1}{4})$.
    \begin{enumerate}
        \item $\fatgamma{\Rc_w} \geq 2$ for $\gamma \in (0, \frac{1}{4})$.
        Take the set $S = \{v_1, v_2\} = \left\{ 4 \gamma, 1 \right\} \subseteq [0,1]$ and the witness $\cb = (c_1, c_2) = \left( \gamma, 3\gamma \right)$. Consider each ${\sigmab} \in \{-1, 1\}^2$:
        \be
        \item ${\sigmab} = (1, 1)$: Take $\p^{\sigmab} \in \npureset$ such that $p^{\sigmab}_w = 4\gamma$, so we have:
        \begin{align*}
            \ind{v_1 \geq p^{\sigmab}_w} \cdot p^{\sigmab}_w &= \ind{4\gamma \geq 4\gamma} \cdot 4 \gamma = 4\gamma \geq 2\gamma = \gamma + \gamma = c_1 + \gamma, \\
            \ind{v_2 \geq p^{\sigmab}_w} \cdot p^{\sigmab}_w &= \ind{1 \geq 4\gamma} \cdot 4\gamma = 4\gamma = 3\gamma + \gamma =  c_2 + \gamma \geq c_2 + \gamma.
        \end{align*}
        
        \item ${\sigmab} = (1, -1)$: Take $\p^{\sigmab} \in \npureset$ such that $p^{\sigmab}_w = 2\gamma$, so we have:
        \begin{align*}
            \ind{v_1 \geq p^{\sigmab}_w} \cdot p^{\sigmab}_w &= \ind{4\gamma \geq 2\gamma} \cdot 2\gamma = 2\gamma = \gamma + \gamma = c_1 + \gamma \geq c_1 + \gamma, \\
            \ind{v_2 \geq p^{\sigmab}_w} \cdot p^{\sigmab}_w &= \ind{1 \geq 2\gamma} \cdot 2\gamma = 2\gamma = 3\gamma - \gamma = c_2 - \gamma \leq c_2 - \gamma.
        \end{align*}
        
        \item ${\sigmab} = (-1, 1)$: Take $\p^{\sigmab} \in \npureset$ such that $p^{\sigmab}_w = 1$, so we have:
        \begin{align*}
            \ind{v_1 \geq p^{\sigmab}_w} \cdot p^{\sigmab}_w &= \ind{4\gamma \geq 1} \cdot 4\gamma = 0 = \gamma - \gamma = c_1 - \gamma  \leq c_1 - \gamma,\\
            \ind{v_2 \geq p^{\sigmab}_w} \cdot p^{\sigmab}_w &= \ind{1 \geq 1} \cdot 1 = 1 \geq 4\gamma = 3\gamma + \gamma = c_2 + \gamma \geq c_2 + \gamma.
        \end{align*}
        
        \item ${\sigmab} = (-1, -1)$: Take $\p^{\sigmab} \in \npureset$ such that $p^{\sigmab}_w =0$, so we have:
        \begin{align*}
            \ind{v_1 \geq p^{\sigmab}_w} \cdot p^{\sigmab}_w &= \ind{4\gamma \geq 0} \cdot 0 = 0 = \gamma - \gamma = c_1 - \gamma \leq c_1 - \gamma, \\
            \ind{v_2 \geq p^{\sigmab}_w} \cdot p^{\sigmab}_w &= \ind{1 \geq 0} \cdot 0 = 0 \leq 2\gamma = 3\gamma - \gamma = c_2 - \gamma \leq c_2 - \gamma,
        \end{align*}
        \ee
        as required.
        
    \item $\fatgamma{\Rc_w} \leq 2$ for $\gamma > 0$. Assume by contradiction that there exists a set $S = \{v_1, v_2, v_3 \} \subseteq [0,1]$, where $v_1 < v_2 < v_3$, which is $\gamma$-shattered by $\Rc_w$ with some witness $\cb \in \R^3$. Using \cref{lem:r-non-neg}, we get that $\cb$ is non-negative. 
    
    First, we assume that $c_2 \geq c_3$. Take ${\sigmab} = (-1, 1, -1)$. By the $\gamma$-shattering assumption, there exists $\p^{\sigmab} \in \npureset$ such that
    \begin{align}
        \ind{v_1 \geq p^{\sigmab}_w} \cdot p^{\sigmab}_w &\leq c_1 - \gamma, \nonumber\\
        \ind{v_2 \geq p^{\sigmab}_w} \cdot p^{\sigmab}_w &\geq c_2 + \gamma, \label{v2+}\\
        \ind{v_3 \geq p^{\sigmab}_w} \cdot p^{\sigmab}_w &\leq c_3 - \gamma. \label{v3-}
    \end{align}
    Since $c_2 \geq 0$ and $\gamma > 0$, we get from (\ref{v2+}) that $\ind{v_2 \geq p^{\sigmab}_w} = 1$, so we have:
    \[
    c_2 + \gamma \leq p^{\sigmab}_w \leq v_2 < v_3.
    \]
    In particular, $v_3 > p^{\sigmab}_w$ and thus $\ind{v_3 \geq p^{\sigmab}_w} = 1$, so from (\ref{v3-}) we get that $p^{\sigmab}_w \leq c_3 - \gamma$. Since $c_2 \geq c_3$, we obtain that
    \[
    c_2 + \gamma \leq p^{\sigmab}_w \leq c_3 - \gamma \leq c_2 - \gamma \Rightarrow \gamma \leq 0,
    \]
    which is a contradiction.
    
    Now, assume that $c_2 < c_3$. Take ${\sigmab} = (1, -1, 1)$. By the $\gamma$-shattering assumption, there exists $\p^{\sigmab} \in \npureset$ such that
    \begin{align}
        \ind{v_1 \geq p^{\sigmab}_w} \cdot p^{\sigmab}_w &\geq c_1 + \gamma, \label{v1+}\\
        \ind{v_2 \geq p^{\sigmab}_w} \cdot p^{\sigmab}_w &\leq c_2 - \gamma, \label{v2-}\\
        \ind{v_3 \geq p^{\sigmab}_w} \cdot p^{\sigmab}_w &\geq c_3 + \gamma. \label{v3+}
    \end{align}
    Since $c_1, c_3 \geq 0$ and $\gamma > 0$, we get from (\ref{v1+}) and (\ref{v3+}) that $\ind{v_1 \geq p^{\sigmab}_w} = \ind{v_3 \geq p^{\sigmab}_w} = 1$, so we have:
    \begin{align*}
        c_1 + \gamma &\leq p^{\sigmab}_w \leq v_1 < v_2,\\
        c_3 + \gamma &\leq p^{\sigmab}_w \leq v_3.
    \end{align*}
    In particular, $v_2 > p^{\sigmab}_w$ and thus $\ind{v_2 \geq p^{\sigmab}_w} = 1$, so from (\ref{v2-}) we get that $p^{\sigmab}_w \leq c_2 - \gamma$. Since $c_2 < c_3$, we obtain:
    \[
    c_2 + \gamma < c_3 + \gamma \leq p^{\sigmab}_w \leq c_2 - \gamma \Rightarrow \gamma \leq 0,
    \]
    which is a contradiction.
    \end{enumerate}
\end{enumerate}

\end{proof}

\thmfatupper*

\begin{proof}
By \cref{lem:fat-gain} we obtain, \[
\int_{0}^{\infty} \sqrt{\fatgamma{\Rpn}} \,d\gamma \leq \int_{0}^{\frac{1}{2}} \sqrt{2\wmax} \,d\gamma = \frac{\sqrt{2\wmax}}{2} = \sqrt{\frac{\wmax}{2}}.
\]
Let $S =\set{z_i}_{i=1}^m$ be a sample of examples drawn i.i.d. according to $\D$. By \cref{uc-theorem} we conclude the claim,

\begin{align*}
    \sup_{f \in \Rpn} \abs{\frac{1}{m} \sum_{i=1}^m f(z_i) - \Eu{z \sim \D}{f(z)}} 
    \lesssim \sqrt{\frac{\wmax +\log \frac{1}{\delta}}{m}}.
\end{align*}
\end{proof}

\thmpuresecondupper*

\begin{proof}
Approximate any decreasing prices with only $k$ different prices (we will set later $k = 1/\eps$). Decrease each price down to the nearest $i/k$, for $0 \leq i \leq k$. The new prices are still decreasing, any buyer that originally buys, still buys. The loss of this approximation is at most $1/k$, due to the fact that we decrease the price by at most $1/k$.

We remain with a limited policy class. There are at most $N=\wmax^k$ different sequences of decreasing prices with only $k$ prices. Indeed, by combinatorial argument (unorder sampling with replacement), the class size is
$\binom{k+\wmax-1}{k}\lesssim \wmax^k$. For learning this class, we can set the approximation error to be the learning error, $k=1/\eps$.
We need a sample of size 
$$m=\frac{\log N}{\eps^2}+\frac{1}{\eps^2}\log\frac{1}{\delta}=\frac{\log(\wmax)}{\eps^3}+\frac{1}{\eps^2}\log\frac{1}{\delta},$$
in order to learn a finite class of size $N$.
\end{proof}

\subsection*{Sample complexity: lower bounds}

\thmlowerboundone*

\begin{proof}
We assume that all buyers have the same patience $w$.
Let $\eps \leq \frac{1}{4}, \delta$ and $m \lesssim \frac{1}{\eps^2} \log \frac{1}{\delta}$. For each $\sigma \in \{-1, +1\}$, define a distribution $\D_{\sigma}$ as follows:
\[
\D_{\sigma}[v,w] = 
\begin{cases}
\frac{1}{2} +\sigma \eps, & v=1 \\
\frac{1}{2} - \sigma \eps, & v=\frac{1}{2}.
\end{cases}
\]

By \cref{lem:pure-over-values}, the seller would prefer to offer either price $1/2$ or price $1$. The seller's excepted revenue for each price $p \in \set{1/2,1}$ and $\sigma \in \{-1, +1\}$  is

\begin{center}
\begin{tabular}{ |c||c|c| } 
 \hline
  & $\sigma = +1$ & $\sigma = -1$\\ [0.5ex]
 \hline\hline
 $p = 1$  & $\frac{1}{2} +\eps$ & $\frac{1}{2} -\eps$\\ [1ex]
 \hline
 $p = \frac{1}{2}$  & $\frac{1}{2}$ & $\frac{1}{2}$\\ [1ex]
 \hline
\end{tabular}
\end{center}

Hence, the optimal price for $\D_{\sigma}$ is
\[
p_{\sigma}^{\star} = 
\begin{cases}
1 & \sigma = +1 \\
\frac{1}{2} & \sigma = -1
\end{cases}
\]
Note that if the seller would offer a price other than $p_{\sigma}^{\star}$ to a buyer from $\D_{\sigma}$, his error is at least $\eps$ in expectation. 

Denote by $p_S$ the price returned by a learning algorithm $\mathcal{A}$ after receiving a sample of size $m$ from distribution $\D_{\sigma}$, where $\sigma$ is sampled uniformly over $\set{-1,+1}$. Now, we can treat our sampling process as an experiment of flipping $m$ times one of two coins with biases $\frac{1}{2} +\eps$ and $\frac{1}{2} -\eps$. 
In order to distinguish which coin was flipped, whether $\D_+$ or $\D_-$, with failure probability of at most $\delta$, a sample size of at least $\approx\frac{1}{\eps^2} \log \frac{1}{\delta}$ is required. Since $m$ is lower than that, with probability of at least $\delta$, algorithm $\mathcal{A}$ would not pick the optimal price, so the seller's error would be at least $\eps$, as required.
\end{proof}

\begin{lemma}\label{lem:prices-when-decreasing-values}
Assume that any two buyer types $(v,w), (v',w')$ in the support $S_{\D}$ of the joint distribution $\D$ hold that if $w' > w$ then $v' < v$. Denote the set of values of buyer types with patience $w \in [\wmax]$ by $\cU_w$. Then, in each step $i$, the seller prefers to offer prices only from $\cU_i$.
\end{lemma}

\begin{proof}
Denote $\cU = \bigcup_{w=1}^T \cU_w$. First, by \cref{lem:pure-over-values}, the set of possible prices $P \subseteq \cU$. We say that step $i$ is a $\textit{violating step}$ in strategy $\p$ if $p_i \notin \cU_i$. Denote the set of all pure selling strategies with no violating steps by $\puresethat$.

We show that for every pure selling strategy $\p \in \pureset \setminus \puresethat$ there exists a pure selling strategy $\hat{\p} \in \puresethat$, whose excepted revenue is greater than or equal to that of $\p$. Assume by contradiction that it does not hold. I.e., there exists a strategy in $\pureset \setminus \puresethat$ that is better than all the strategies in $\puresethat$. Of all the \textit{optimal} strategies, all of which are necessarily in $\pureset \setminus \puresethat$, consider a strategy $\p$ where the first violating step $s$ in $\p$ is the latest. We have $p_s \in P \setminus{\cU_s}$. By \cref{lem:non-increasing-pure}, we can assume without loss of generality that $\p$ is a non-increasing pricing, so the buyer would buy exactly in the step in which his patience expires. For every $i \in [\wmax]$, denote $\ov_i = \max_{v \in \cU_i} v, \uv_i = \min_{v \in \cU_i} v$. Consider the following cases:

\bi

\item $p_s > \ov_s$: Let $s' \geq s$ be the \textit{last} violating step such that for all $i \in [s, s']$, $p_i > \ov_i$. Every buyer type with patience $\in [s, s']$ cannot buy the item. Thus, the revenues of $\p$ in those steps are zero. Define strategy $\p'$ to be the same as $\p$ except that we replace $p_i$ with $\ov_i$ for all $i \in [s, s']$. Since $s$  is the first violating step, $p_{s-1} \in \cU_{s-1}$. Hence, since the values decrease over time, we have:
\begin{align*}
    &p'_{s-1} = p_{s-1} > \ov_s = p'_s,\\
    &p'_i = \ov_i > \ov_{i+1} = p'_{i+1} \quad \forall i \in [s, s'),\\
    &p'_{s'} = \ov_{s'} > \ov_{s'+1} \geq p_{s'+1} = p'_{s'+1}.
\end{align*}
Therefore, $\p'$ is also non-increasing strategy and its excepted revenue is greater than or equal to that of $\p$.

\item $p_s < \uv_s$: There may be buyer types with patience $s$ who buy the item in this step. Define strategy $\p'$ to be the same as $\p$ except that we replace $p_s$ with $\uv_s$. Since this is their final step, every buyer type with patience $s$ buys the item at this step at price $p'_s = \uv_s > p_s$. Since $s$ is the first violating step, $p_{s-1} \in \cU_{s-1}$. Hence, since the values decrease over time and $\p$ is a non-increasing pricing, we have:
\begin{align*}
    &p'_{s-1} > \uv_s = p'_s,\\
    &p'_s = \uv_s > p_s \geq p_{s+1} = p'_{s+1}.
\end{align*}
Therefore, $\p'$ is also non-increasing strategy and its excepted revenue is greater than or equal to that of $\p$.
\ei

In both cases, we define a non-increasing strategy $\p'$ whose excepted revenue is at least as much as $\p$'s. In particular, $\p'$ is an optimal strategy. However, if $\p' \in \puresethat$, we get a contradiction to the assumption that there are no optimal strategies in $\puresethat$; otherwise, since the first violating step of $\p'$ is later than that of $\p$, we get a contradiction to the assumption that of all the optimal strategies in $\pureset \setminus \puresethat$, the first violating step in $\p$ is the latest.
\end{proof}

\thmlowerboundtwo*

\begin{proof}
Formally, we show that for any learning algorithm $\cA:\Zc^{*}\rightarrow [0,1]^{\wmax}$, there exists a distribution $\D$ over $[0,1]\times[\wmax]$ such that for a random sample $S\sim \D^m$, for $m\leq c_1\frac{\wmax}{\eps}$, $\eps\leq c_2\frac{1}{\wmax}$ and $\delta \leq c_3$, it holds that
$$
\Probu{S\sim \D^m}{ \rev(\cA(S);\D) < \max_{\p\in [0,1]^{\wmax}}\rev(\p;\D)-\eps}\geq \delta,
$$
where $c_1,c_2,c_3 > 0$ are universal constants.
For each patience $w \in [\wmax]$ where $\wmax\geq 2$, define the set
\[
\V_w = \left\{ v_1^w, v_2^w \right\} = \left\{ \left( \frac{\wmax-1}{\wmax} \right)^{2w-1}, \; \left( \frac{\wmax-1}{\wmax} \right)^{2w} \right\}.
\]

Define a set $\Zc = \bigcup_{w=1}^{\wmax} \V_w \times \{w\} \subseteq [0,1] \times [\wmax]$. Denote $\alpha = \frac{\eps}{16}$. For any $\sigmab \in \{-1, +1\}^{\wmax}$, define a distribution $\D_{\sigmab}$ over $\Zc$ as follows:
\[
\forall w \in [\wmax]. \quad \D_{\sigmab}[v, w] =
\begin{cases}
\frac{1}{\wmax} \cdot q (1 + \sigma_{w}\alpha), & v=v^w_1 \\
\frac{1}{\wmax} \cdot \paren{1- q (1 +\sigma_{w}\alpha)}, & v=v^w_2,
\end{cases}
\]

where $q = v_2^w/v_1^w=(\wmax-1)/\wmax$ which is in $[0,1]$. 
Sampling according to $\D_{\sigmab}$ can be described as follows. First we sample $w \in [\wmax]$ uniformly at random, and then set the value to be $v_1^w$ with probability $q (1 + \sigma_{w}\alpha)$ or $v_2^w$ with probability $1- q (1 +\sigma_{w}\alpha)$. 

Next, we find the Bayes optimal hypothesis for $\D_{\sigmab}$. 
Any distribution $\D_{\sigmab}$ satisfies the conditions of  \cref{lem:prices-when-decreasing-values}, and we can deduce that the Bayes optimal offers prices only from $\V_i$ at step $i \in [\wmax]$. Moreover, for any $i\in[\wmax-1]$, the values at step $i$ are greater than the values at step $i+1$, that is, the Bayes optimal is a non-increasing pure selling strategy.

We show that the optimal price at step $i$ is determined by the values of $\sigma_i\in\set{-1,1}$.
The seller's excepted revenue at step $i \in [\wmax]$, for each price $p_i \in \V_i$ and $\sigma_i \in \{-1, +1\}$  is

\begin{center}
\begin{tabular}{ |c||c|c| } 
 \hline
  & $\sigma_i = +1$ & $\sigma_i = -1$\\ [0.5ex]
 \hline\hline
 $p_i = v_1^i$  & $v_1^i \cdot q (1 + \alpha)$ & $v_1^i \cdot q (1 -  \alpha)$\\ [1ex]
 \hline
 $p_i = v_2^i$  & $v_2^i$ & $v_2^i$\\ [1ex]
 \hline
\end{tabular}
\end{center}
Note that,
\begin{align} 
    \label{eq:loss-in-round}
    v_1^i q (1 + \alpha) - v_2^i =  v_1^i \cdot \frac{v_2^i}{v_1^i} (1 +\alpha) - v_2^i = v_2^i \alpha,
    \\
    \label{eq:loss-in-round2}
    v_2^i - v_1^i q (1 - \alpha) =  v_2^i - v_1^i \cdot \frac{v_2^i}{v_1^i} (1 - \alpha) = v_2^i \alpha.
\end{align}
Therefore, the Bayes optimal hypothesis for $\D_{\sigmab}$ is a non-increasing selling strategy and defined as, 
\[
\pds(i) = 
\begin{cases}
v_1^i & \sigma_i = +1 \\
v_2^i & \sigma_i = -1.
\end{cases}
\]
It follows from \cref{eq:loss-in-round,eq:loss-in-round2} that a selling strategy offering a different price than the Bayes optimal at step $i$, resulting in a lower revenue from buyers at step $i$ by $v_2^i\alpha$, for $i\in[\wmax]$. We claim that for any $i$, it holds that $v_2^i\geq\frac{1}{16}$, and as a result we have $v_2^i \alpha \geq \frac{1}{16} \alpha = \eps$.

The fact that $v_2^i\geq\frac{1}{16}$ is derived from the following. $v^{\wmax}_2$ is the smallest element in the set $\set{v^1_2,v^2_2,\ldots,v^{\wmax}_2}$, and $v^{\wmax}_2\geq \frac{1}{16}$ for $\wmax\geq 2$, since the sequence, $(v^{\wmax}_2)^{\infty}_{\wmax=2}$ is monotonic increasing.

Denote by $\pbf_S$ the hypothesis returned by the learning algorithm $\mathcal{A}$ upon receiving a sample $S$ of size $m$, drawn according to $S\sim\D^m_{\sigmab}$, where $\sigmab$ is sampled uniformly over $\set{-1,+1}^{\wmax}$. Since we have a sample $S$ of size $m \lesssim \wmax/\eps$, there exists at least a constant fraction of patience windows $w$ such that the number of buyer types with patience $w$ in $S$ is less than $1/\eps$. However, for such $w$, the probability of $v_2^w$ is about $\eps$, so with non-negligible probability, $\sigma_w = -1$ and $v_2^w$ is not observed, which leads to a loss of at least $\eps$ for each such $w$.
\end{proof}

\subsection{Proofs for \cref{subsec:sample-complexity-mix}}\label{app:sample-complexity-mix}

\thmsamplecomplexitymixone*
\begin{proof}
The sample complexity of learning discrete distributions over a known domain of size $k\wmax$, with respect to the total variation distance, is $\Theta\left(\frac{k\wmax+\log\frac{1}{\delta}}{\eps^2} \right)$. See for example, \citet[Theorem 1]{canonne2020short},
this is also known as Bretagnolle Huber-Carol inequality. Denote the approximated distribution by $\hat{\D}$. Since the revenue of any strategy is at most $1$, we conclude,
\begin{align*}
r(\Pd;\D)-r(\Pd;\hat{\D})
&=
\frac{1}{2}
\sum_{(v,w)}\left(
\D[v,w]-\hat{\D}[v,w]\right)r(\Pd,(v,w)) \\
&\leq
\frac{1}{2}
\sum_{(v,w)}\left(
\D[v,w]-\hat{\D}[v,w]\right)\\
&\leq \eps.
\end{align*}

\end{proof}

\thmsamplecomplexitymixtwo*
\begin{proof}
From \cref{lem:eps-w-optimal-mixed}, we get an $\eps\wmax$-optimal strategy. By plugging it in the aforementioned proof, the claim follows.
\end{proof}

\subsection{Proofs for \cref{subsec:regret}}\label{app:regret}

\thmregretpure*
\begin{proof}
We describe an algorithm $\cA_{ERM}$ and its guarantees. At times $2^i$ for $1\leq i \leq  \log T$, we call an $\ERM$ on entire sequence $z_1,\ldots,z_t$, assuming that $T$ is a power of $2$. 
Each time we invoke the $\ERM$, we have an error of at most $\eps$ with probability $1-\delta'$. We make $\log T$ calls, and the failure probability grows by factor $\log T$, by taking $\delta'=\delta/\log T$ we get the following regret bound with probability $1-\delta$. We calculate the regret of $\cA_{ERM}$ for playing strategies $\p_1, \ldots, \p_T$,

\begin{align*}
\regret_T^{\pureset} \big(\cA_{\ERM};\D\big)
&=
\max_{\p^{\star}\in \pureset}\sum_{t=1}^T 
\Eu{z_t\sim \D}{r(\p^{\star};z_t) - r(\p_t;z_t)}  
\\
&=
\sum_{t=1}^{T} 
\left[ \min \left\{
\frac{\log^{1/3}(\wmax)}{t^{1/3}}\;,\;
\frac{\sqrt{\wmax}}{t^{1/2}}\right\}+\frac{1}{t^{1/2}}\sqrt{\log\frac{\log T}{\delta}} \right]
\\
&\lesssim
\sum_{i=0}^{\log T} 2^i
\left[\min \left\{
\frac{\log^{1/3}(\wmax)}{2^{i/3}}\;,\;
\frac{\sqrt{\wmax}}{2^{i/2}}\right\}+\frac{1}{2^{i/2}}\sqrt{\log\frac{\log T}{\delta}}\right]
\\
&=
\sum_{i=0}^{\log T} 2^i
\cdot \min \left\{
\frac{\log^{1/3}(\wmax)}{2^{i/3}}\;,\;
\frac{\sqrt{\wmax}}{2^{i/2}}\right\}
+
\sum_{i=0}^{\log T} 2^{i/2}\sqrt{\log\frac{\log T}{\delta}}
\\
&\lesssim
\min \left\{ T^{2/3}\log^{1/3}(\wmax),\sqrt{T}\sqrt{\wmax} \right\}
+\sqrt{T}\sqrt{\log\frac{\log T}{\delta}}.
\end{align*}
\end{proof}

\thmregretmixone*
\begin{proof}
We describe an algorithm $\cA_{ERM}$ and its guarantees. At times $2^i$ for $1\leq i \leq  \log T$, we call an $\ERM$ on entire sequence $z_1,\ldots,z_t$, assuming that $T$ is a power of $2$. 
Each time we invoke the $\ERM$, we have an error of at most $\eps$ with probability $1-\delta'$. We make $\log T$ calls, and the failure probability grows by factor $\log T$, by taking $\delta'=\delta/\log T$ we get the following regret bound with probability $1-\delta$. We calculate the regret of $\cA_{ERM}$ for playing strategies $\Pd_1, \ldots, \Pd_T$,

\begin{align*}
\regret_T^{\mixedsetV} \big(\cA_{\ERM};\D\big)
&=
\max_{\Pdst\in\mixedsetV}\sum_{t=1}^T 
\Eu{z_t\sim \D}{r(\Pdst;z_t) - r(\Pd_t;z_t)}  
\\
&=
\sum_{t=1}^{T} 
\frac{\sqrt{|\V|\cdot\wmax}}{t^{1/2}}+\frac{1}{t^{1/2}}\sqrt{\log\frac{\log T}{\delta}}
\\
&\lesssim
\sum_{i=0}^{\log T} 2^i
\left[
\frac{\sqrt{|\V|\cdot\wmax}}{2^{i/2}}+\frac{1}{2^{i/2}}\sqrt{\log\frac{\log T}{\delta}}\right]
\\
&=
\sum_{i=0}^{\log T} 2^{i/2}
\sqrt{|\V|\cdot\wmax}
+
\sum_{i=0}^{\log T} 2^{i/2}\sqrt{\log\frac{\log T}{\delta}}
\\
&\lesssim
\sqrt{T}\sqrt{|\V|\cdot\wmax}
+\sqrt{T}\sqrt{\log\frac{\log T}{\delta}}.
\end{align*}
\end{proof}

\thmregretmixtwo*

\begin{proof}
We describe an algorithm $\cA_{ERM}$ and its guarantees. At times $2^i$ for $1\leq i \leq  \log T$, we call an $\ERM$ on entire sequence $z_1,\ldots,z_t$, assume that $T$ is a power of $2$. 
Each time we invoke the $\ERM$, we have an error of at most $\eps$ with probability $1-\delta'$. We make $\log T$ calls, and the failure probability grows by factor $\log T$, by taking $\delta'=\delta/\log T$ we get the following regret bound with probability $1-\delta$. We calculate the regret of $\cA_{ERM}$ for playing strategies $\Pd_1, \ldots, \Pd_T$,

\begin{align*}
\regret_T^{\mixedset} \big(\cA_{\ERM};\D\big)
&=
\max_{\Pdst\in\mixedset}\sum_{t=1}^T 
\Eu{z_t\sim \D}{r(\Pdst;z_t) - r(\Pd_t;z_t)}  
\\
&=
\sum_{t=1}^{T} 
\left[\frac{\wmax^{4/3}}{t^{1/3}}+\frac{\wmax}{t^{1/2}}\sqrt{\log\frac{\log T}{\delta}}\right]
\\
&\lesssim
\sum_{i=0}^{\log T} 2^i
\left[
\frac{\wmax^{4/3}}{2^{i/3}}+\frac{\wmax}{2^{i/2}}\sqrt{\log\frac{\log T}{\delta}}\right]
\\
&=
\sum_{i=0}^{\log T} \left[2^{2i/3}
\wmax^{4/3}
+2^{i/2}\wmax\sqrt{\log\frac{\log T}{\delta}}\right]
\\
&\lesssim
T^{2/3}\wmax^{4/3}
+\sqrt{T}\wmax\sqrt{\log\frac{\log T}{\delta}}.
\end{align*}
\end{proof}

\end{document}

%% file: header.tex
\usepackage{amsthm} 
\usepackage{mathtools,thmtools}

\usepackage{enumitem}
\usepackage{bm}
\usepackage{xspace}
\usepackage{nicefrac}

\usepackage[capitalise]{cleveref}
\usepackage{dsfont}

\usepackage[usenames,dvipsnames]{xcolor}

\declaretheoremstyle[
	    spaceabove=\topsep, 
	    spacebelow=\topsep, 
	    headfont=\normalfont\bfseries,
	    bodyfont=\normalfont\itshape,
	    notefont=\normalfont\bfseries,
	    notebraces={(}{)},
	    postheadspace=0.5em, 
	    headpunct={},
    ]{theorem}
\declaretheorem[style=theorem,numberwithin=section]{theorem}

\declaretheoremstyle[
	    spaceabove=\topsep, 
	    spacebelow=\topsep, 
	    headfont=\normalfont\bfseries,
	    bodyfont=\normalfont,
	    notefont=\normalfont\bfseries,
	    notebraces={(}{)},
	    postheadspace=0.5em, 
	    headpunct={},
    ]{definition}

\declaretheoremstyle[
        spaceabove=\topsep, 
        spacebelow=\topsep, 
        headfont=\normalfont\bfseries,
        bodyfont=\normalfont,
        notefont=\normalfont\bfseries,
        notebraces={}{},
        postheadspace=0.5em, 
        qed=$\blacksquare$, 
        headpunct={},
    ]{proofstyle}
\declaretheorem[style=proofstyle,numbered=no,name=Proof]{proof}

\declaretheorem[style=theorem,sibling=theorem,name=Lemma]{lemma}
\declaretheorem[style=theorem,sibling=theorem,name=Corollary]{corollary}

\declaretheorem[style=theorem,sibling=theorem,name=Proposition]{proposition}

\declaretheorem[style=theorem,numbered=no,name=Theorem]{theorem*}
\declaretheorem[style=theorem,numbered=no,name=Lemma]{lemma*}
\declaretheorem[style=theorem,numbered=no,name=Corollary]{corollary*}
\declaretheorem[style=theorem,numbered=no,name=Proposition]{proposition*}
\declaretheorem[style=theorem,numbered=no,name=Claim]{claim*}
\declaretheorem[style=theorem,numbered=no,name=Fact]{fact*}
\declaretheorem[style=theorem,numbered=no,name=Observation]{observation*}
\declaretheorem[style=theorem,numbered=no,name=Conjecture]{conjecture*}

\declaretheorem[style=definition,sibling=theorem,name=Definition]{definition}
\declaretheorem[style=definition,sibling=theorem,name=Remark]{remark}

\declaretheorem[style=definition,numbered=no,name=Definition]{definition*}
\declaretheorem[style=definition,numbered=no,name=Remark]{remark*}
\declaretheorem[style=definition,numbered=no,name=Example]{example*}
\declaretheorem[style=definition,numbered=no,name=Question]{question*}

\DeclareMathAlphabet{\mathbfsf}{\encodingdefault}{\sfdefault}{bx}{n}

\DeclareMathOperator*{\argmax}{arg\!\max}

\let\Pr\relax
\DeclareMathOperator{\Pr}{\mathbb{P}}

\newcommand{\abs}[1]{\left| #1 \right|}

\newcommand{\ind}[1]{\mathbb{I}\set{#1}}

\renewcommand{\O}{\mathcal{O}}

\newcommand{\poly}{\mathrm{poly}}

\newcommand{\eps}{\varepsilon}

\let\oldtfrac\tfrac
\renewcommand{\tfrac}[2]{\smash{\oldtfrac{#1}{#2}}}

\let\nablaold\nabla
\renewcommand{\nabla}{\nablaold\mkern-1mu}
\crefname{observation}{Observation}{Observations}

\newcommand{\paren}[1]{\left( #1 \right)}

\usepackage{amsfonts,amsmath,amssymb}
\usepackage{mathtools}
\usepackage{float}
\usepackage{enumitem}

\newcommand{\cR}{\mathcal{R}}

\newcommand{\cO}{\mathcal{O}}

\newcommand{\cM}{\mathcal{M}}

\newcommand{\cA}{\mathcal{A}}

\newcommand{\eeq}{\end{aligned}\end{equation}}
\newcommand{\beq}{\begin{equation}\begin{aligned}}

\usepackage{amsfonts,amsmath,amssymb,amsthm,boxedminipage,color,url,fullpage}
\usepackage{graphicx}
\usepackage{mathabx}
\usepackage[labelfont=bf]{caption}
\usepackage{aliascnt,cleveref}
\usepackage{enumitem}
\usepackage{dsfont}
\RequirePackage{algorithm}
\usepackage{pgfplots}
\pgfplotsset{compat=1.17} 

\newcommand{\bi}{\begin{itemize}}
\newcommand{\ei}{\end{itemize}}
\newcommand{\be}{\begin{enumerate}}
\newcommand{\ee}{\end{enumerate}}
\newcommand{\bea}{\begin{enumerate}[label=(\alph*)]}
\newcommand{\eea}{\end{enumerate}}

\newcommand{\remove}[1]{}
\newcommand{\Draft}[1]{\ifnum\draft=1\texttt{ #1} \fi}

\newcommand{\set}[1]{\ens{#1}}

\newcommand{\R}{{\mathbb R}}
\newcommand{\N}{{\mathbb{N}}}

\renewcommand{\cref}{\Cref}

\newcommand{\ens}[1]{\left\{#1\right\}}

\newcommand{\cb}{\mathbf{c}}

\makeatletter
\newcommand{\subalign}[1]{%
  \vcenter{%
    \Let@ \restore@math@cr \default@tag
    \baselineskip\fontdimen10 \scriptfont\tw@
    \advance\baselineskip\fontdimen12 \scriptfont\tw@
    \lineskip\thr@@\fontdimen8 \scriptfont\thr@@
    \lineskiplimit\lineskip
    \ialign{\hfil$\m@th\scriptstyle##$&$\m@th\scriptstyle{}##$\hfil\crcr
      #1\crcr
    }%
  }%
}
\makeatother

\newcommand{\Prob}[1]{\Pr\left( #1 \right)}
\newcommand{\Probb}[2]{\Pr_{#1} \left( #2 \right)}
\newcommand{\Probu}[2]{\underset{#1}{\Pr} \left( #2 \right)}

\newcommand{\E}[1]{{\mathbb E}\left[ #1 \right]}
\newcommand{\Eb}[2]{{\mathbb E}_{#1} \left[ #2 \right]}
\newcommand{\Eu}[2]{\underset{#1}{\mathbb{E}} \left[ #2 \right]}
\newcommand{\Eun}[2]{\underset{#1}{\mathbb{E}} #2 }

\newcommand{\D}{\mathcal{D}}
\newcommand{\V}{\mathcal{V}}

\newcommand{\F}{\mathcal{F}}
\newcommand{\W}{\mathcal{W}}
\newcommand{\Pd}{\mathcal{P}}

\newcommand{\Zc}{\mathcal{Z}}

\newcommand{\bbf}{\mathbf{b}}

\newcommand{\pbf}{\mathbf{p}}
\newcommand{\p}{\mathbf{p}}
\newcommand{\pj}[1]{\p_{1:#1}}

\newcommand{\e}{\mathbf{e}}
\newcommand{\sigmab}{\boldsymbol{\sigma}}

\newcommand{\pistP}{\pi^{\star}_{v,w,\Pd}}

\newcommand{\Pdst}{\Pd^{\star}}
\newcommand{\pst}{\p^{\star}}
\newcommand{\ppst}{p^{\star}}
\newcommand{\bb}{\mathbf{b}}

\newcommand{\Ptilde}{\tilde{\Pd}}

\newcommand{\bbtilde}{\tilde{\bb}}
\newcommand{\rtilde}{\tilde{r}}
\newcommand{\alphatilde}{\tilde{\alphabf}}

\newcommand{\alphabf}{\boldsymbol{\alpha}}

\newcommand{\zb}{\mathbf{z}}
\newcommand{\z}{\mathbf{z}}

\newcommand{\cont}{\textit{continue}}

\newcommand{\bbhat}{\hat{\bbf}}

\newcommand{\alphahat}{\hat{\alphabf}}
\newcommand{\phat}{\hat{\p}}

\newcommand{\pitilde}{\tilde{\pi}_{v,w,\Pd}}
\newcommand{\pibar}{\bar{\pi}_{v,w,\Pd}}

\newcommand{\pihat}{\hat{\pi}_{v,w,\Pd}}

\newcommand{\piP}{\pi_{v,w, \Pd}}

\newcommand{\utilP}{\util_{v,w,\Pd}}
\newcommand{\utils}{\util_{v,w}}
\newcommand{\utiltags}{\util_{v',w'}}

\newcommand{\qs}{q_{v,w}}
\newcommand{\mus}{\mu_{v,w}}

\newcommand{\thetahat}{\hat{\theta}_{v,w,\Pd}}
\newcommand{\thetahats}{\hat{\theta}_{v,w}}
\newcommand{\thetahattag}{\hat{\theta}_{v',w',\Pd}}
\newcommand{\thetahattags}{\hat{\theta}_{v',w'}}

\newcommand{\fatgamma}[1]{\fat_{\gamma}{\left( #1 \right)}}

\newcommand{\pds}{\p^{\star}_{\D_{\sigmab}}}

\newcommand{\ov}{\overline{v}}
\newcommand{\uv}{\underline{v}}
\newcommand{\vmax}{\overline{v}}
\newcommand{\vmin}{\underline{v}}

\newcommand{\fat}{\operatorname{fat}}
\newcommand{\pdim}{\operatorname{Pdim}}

\newcommand{\wmax}{\hat{w}}

\newcommand{\Rpn}{\mathcal{R}^{\text{pure}}_{\geq}}

\newcommand{\ERM}{\operatorname{ERM}}

\newcommand{\rev}{r}
\newcommand{\revdec}{\text{rev}}
\newcommand{\util}{u}
\newcommand{\utildec}{\text{util}}

\newcommand{\regret}{\text{Regret}}

\newcommand{\MDv}{M_{\D_v}}
\newcommand{\MP}{M_{\Pd}}

\newcommand{\pureset}{\Omega^{\text{pure}}}
\newcommand{\npureset}{\Omega^{\text{pure}}_{\geq}}
\newcommand{\mixedset}{\Omega^{\text{mixed}}}
\newcommand{\mixedsetV}{\Omega_{\V}^{\text{mixed}}}

\newcommand{\puresethat}{\hat{\Omega}^{\text{pure}}}

\newcommand{\Phat}{\hat{\Pd}}
\newcommand{\Rc}{\mathcal{R}}
\newcommand{\cU}{\mathcal{U}}

    \makeatletter
\def\@fnsymbol#1{\ensuremath{\ifcase#1\or \dagger\or \ddagger\or
   \mathsection\or \mathparagraph\or \|\or **\or \dagger\dagger
   \or \ddagger\ddagger \else\@ctrerr\fi}}
    \makeatother